%% file: arxiv_main.tex
\documentclass{article}
\usepackage{fullpage}

\usepackage{amsmath, amssymb,bm}
\usepackage{hyperref}
\usepackage{enumerate,mathtools}
\usepackage{amsthm}
\usepackage{cite}
\usepackage{microtype}
\usepackage{color}
\usepackage{subcaption}

  \theoremstyle{plain}
  \newtheorem{theorem}{Theorem}
  \newtheorem{lemma}[theorem]{Lemma}

  \theoremstyle{definition}
  \newtheorem{definition}[theorem]{Definition}

  \theoremstyle{remark}

\usepackage{amsmath, amssymb,bm,mathtools,url}

\usepackage{hyperref}

\usepackage{tikz}
\usetikzlibrary{shapes.geometric, shapes.arrows, arrows, positioning}
\tikzstyle{box}=[draw, fill=blue!20, scale= 0.8, minimum size=2em]
\tikzstyle{circ}=[draw, circle, fill=red!20, scale= 0.8, minimum size=2em]
\tikzstyle{circ_blue}=[draw, circle, fill=blue!20, scale= 0.8, minimum size=2em]
\tikzstyle{box_red}=[draw, fill=red!20, minimum size=2em, scale= 0.8]

\usetikzlibrary{arrows, decorations.markings}

\tikzstyle{vecArrow} = [thick, decoration={markings,mark=at position
   1 with {\arrow[semithick]{open triangle 60}}},
   double distance=1.4pt, shorten >= 5.5pt,
   preaction = {decorate},
   postaction = {draw,line width=1.4pt, white,shorten >= 4.5pt}]

\tikzstyle{innerWhite} = [semithick, white,line width=1.4pt, shorten >= 4.5pt]

\bf
\newcommand{\bA}{\mathbf{A}}
\newcommand{\bB}{\mathbf{B}}

\newcommand{\bQ}{\mathbf{Q}}
\newcommand{\bR}{\mathbf{R}}

\newcommand{\bU}{\mathbf{U}}
\newcommand{\bV}{\mathbf{V}}
\newcommand{\bW}{\mathbf{W}}
\newcommand{\bX}{\mathbf{X}}
\newcommand{\bY}{\mathbf{Y}}
\newcommand{\bZ}{\mathbf{Z}}

\newcommand{\bu}{\mathbf{u}}
\newcommand{\bv}{\mathbf{v}}

\newcommand{\bx}{\mathbf{x}}
\newcommand{\by}{\mathbf{y}}
\newcommand{\bz}{\mathbf{z}}

\newcommand{\cF}{\mathcal{F}}

\newcommand{\cI}{\mathcal{I}}

\newcommand{\cM}{\mathcal{M}}

\newcommand{\cX}{\mathcal{X}}

\global\long\def\dd{\mathrm{d}}

\newcommand{\one}{\boldsymbol{1}} 
\newcommand{\goto}{\rightarrow} 

\newcommand{\bEx}{\ensuremath{\mathbb{E}}}

\newcommand{\ex}[1]{\ensuremath{\mathbb{E}\left[ #1\right]}}

\newcommand{\pr}[1]{\ensuremath{\mathbb{P}\left[ #1\right]}}

\DeclareMathOperator{\cov}{\mathsf{Cov}}
\DeclareMathOperator{\mmse}{\sf mmse}
\DeclareMathOperator{\var}{\sf Var}
\DeclareMathOperator{\snr}{\sf snr}

\DeclareMathOperator{\BG}{\mathsf{BG}}

\newcommand{\DKL}[2]{\ensuremath{D\left( #1 \, \|  \, #2 \right)}}

\DeclareMathOperator{\gtr}{tr}

\newcommand{\reals}{\mathbb{R}}

\newcommand{\eps}{\epsilon}
\newcommand{\normal}{\mathcal{N}}
\newcommand{\gmid}{\! \mid \!}

\newcommand{\Id}{\mathbf{I}}

\let\originalleft\left
\let\originalright\right
\renewcommand{\left}{\mathopen{}\mathclose\bgroup\originalleft}
\renewcommand{\right}{\aftergroup\egroup\originalright}

\allowdisplaybreaks

\title{Understanding Phase Transitions via\\ Mutual Information and MMSE} 

\author{Galen Reeves \and Henry Pfister
\thanks{G.~Reeves is with the Department of Electrical and Computer Engineering and the Department of Statistical Science, Duke University, Durham, NC 27708 USA (e-mail: galen.reeves@duke.edu). H.~Pfister is with the Department of Electrical and Computer Engineering, Duke University, Durham, NC 27708 USA (e-mail: henry.pfister@duke.edu).}
}
\begin{document}

 \maketitle


\begin{abstract}
The ability to understand and solve high-dimensional inference problems is essential for modern data science.
This article examines high-dimensional inference problems through the lens of information theory and focuses on the standard linear model as a canonical example that is both rich enough to be practically useful and simple enough to be studied rigorously.
In particular, this model can exhibit phase transitions where an arbitrarily  small change in the model parameters can induce large changes in the quality of estimates. 
For this model, the performance of optimal inference can be studied using the replica method from statistical physics but, until recently, it was not known if the resulting formulas were actually correct.
In this chapter, we present a tutorial description of the standard linear model and its connection to information theory.
We also describe the replica prediction for this model and outline the authors' recent proof that it is exact.
\end{abstract}

\setcounter{tocdepth}{2}
  \tableofcontents

\input{section_intro}
\input{section_slm}
\input{section_role_of_MI} 
\input{section_proof_of_formulas}

\input{section_phase_transitions}

\section{Conclusion} 
\label{sec:conclusion}

This article provides a tutorial introduction to high-dimensional inference and its connection to information theory.
The standard linear model is analyzed in detail and used as a running example.
The primary goal is to present intuitive links between phase transitions, mutual information, and estimation error.
To that end, we show how general functional properties (e.g., the chain rule, data-processing inequality, and I-MMSE relationship) of mutual information and MMSE can imply meaningful constraints on the solutions of challenging problems.
In particular, the replica prediction of the mutual information and MMSE is described and an outline is given for the authors' proof that it is exact in some cases.
We hope that the approach described here will allow this material to be accessible by a wider audience.

\subsection{Further Directions}
Beyond the standard linear model, there are other interesting high-dimensional inference problems that can be addressed using the ideas in this chapter. For example, one recent line of work has focused on multilayer networks, which consist of multiple stages of a linear transform followed by a non-linear (possibly random) function~\cite{manoel:2017a, fletcher:2018a, reeves:2017e}. There has also been significant work on bilinear estimation problems,  matrix factorization, and community detection~\cite{barbier:2016a,lelarge:2016,lesieur:2017, abbe:2017}. Finally, there has been some initial progress on optimal quantization for the standard linear model~\cite{kipnis:2017,kipnis:2018a}.

\input{section_extensions}

\bibliographystyle{IEEEtran} 

\input{chapter_bib.bbl}


\end{document}

%% file: section_intro.tex

\section{Introduction}
\subsection{What Can We Learn From Data?}

Given a probabilistic model, this question can be answered succinctly in terms of the difference  between the prior distribution (what we know before looking at the data)  and the posterior distribution (what we know after looking at the data). Throughout this chapter we will focus on high-dimensional inference problems where the posterior distribution may be complicated and difficult to work with directly. We will show how techniques rooted in information theory and and statistical physics can provide explicit characterizations of the statistical relationship between the data and the unknown quantities of interest. 

In his seminal paper, Shannon showed that \emph{mutual information} provides an important measure of the difference between the prior and posterior distributions~\cite{shannon:1948}. Since its introduction, mutual information has played a central role for  applications in engineering, such as communication and data compression, by describing  the fundamental constraints imposed solely by the \emph{statistical} properties of the problems. 

Traditional problems in information theory assume that all statistics are known and that certain system parameters can be chosen to optimize performance.
In contrast, data-science problems typically assume that the important distributions are either given by the problem or must be estimated from the data.
When the distributions are unknown, the implied inference problems are more challenging and their analysis can become intractable.
Nevertheless, similar behavior has also been observed in more general high-dimensional inference problems such as Gaussian mixture clustering~\cite{lesieur:2016}.

In this chapter, we use the standard linear model as a simple example to illustrate phase transitions in high-dimensional inference.
In Section~\ref{sec:setup}, basic properties of the standard linear model are described and examples are given to describe its behavior.
In Section~\ref{sec:role_of_MI}, a number of connections to information theory are introduced.
In Section~\ref{sec:proof}, we present an overview of the authors' proof that the replica formula for mutual information is exact.
In Section~\ref{sec:phase}, connections between posterior correlation and phase transitions are discussed.
Finally, in Section~\ref{sec:conclusion}, we offer concluding remarks.

\subsubsection{Notation}
The probability $\mathbb{P}[\bY=\by|\bX=\bx]$
is denoted succinctly by $p_{\bY|\bX} (\by|\bx)$ and shortened to $p(\by|\bx)$ when the meaning is clear.
Similarly, the distinction between discrete and continuous distributions is neglected when it is inconsequential. 

\subsection{High-Dimensional Inference} 

Suppose that the relationship between the unobserved random vector  $\bX = (X_1, \dots , X_N)$ and the observed random vector $\bY  = (Y_1, \dots, Y_M)$ is modeled by the joint probability distribution $p (\bx, \by)$. The central problem of Bayesian inference is to answer questions about the unobserved variables in terms of the posterior distribution:
\begin{align}
p(\bx \gmid \by)  =  \frac{ p(\by \gmid \bx) p(\bx) }{ p(\by)} .
\end{align}
In the high-dimensional setting where both $M$ and $N$ are large, direct evaluation of the posterior distribution can become intractable and one often resorts to summary statistics such as the posterior mean/covariance or the marginal posterior distribution of a small subset of the variables. 

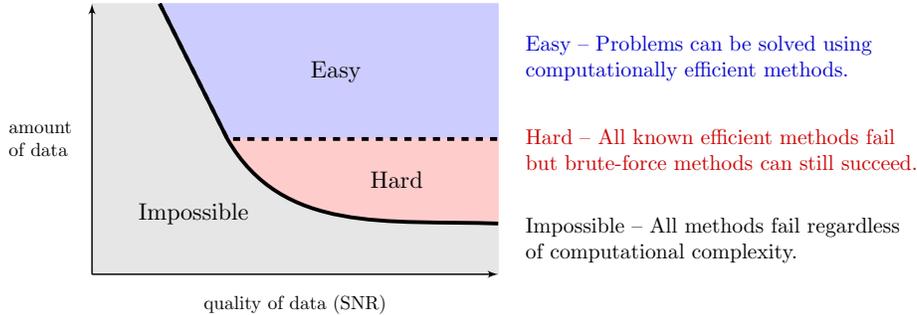
\begin{figure}[]
\centering
\scalebox{.9}{
\begin{tikzpicture}[>= latex']

\coordinate (A) at (1,4);
\coordinate (B) at (2,2);
\coordinate (C) at (6,2);
\coordinate (D) at (6,.75);

\draw[thin, fill,  color = black!10] (0,0) -- (0,4) -- (A) -- (B) to[out= -60,in=178] (D) -- (6,0) -- (0,0);
\draw[thin, fill,  color = blue!20]  (6,4) -- (A) -- (B) -- (C)  -- (6,4);
\draw[thin, fill,  color = red!20]  (B) to[out= -60,in=178] (D)-- (C)  -- (B);

\draw[semithick, ->] (0,0) -- coordinate (x axis mid) (6,0);
\draw[semithick, ->] (0,0) -- coordinate (y axis mid) (0,4);

\node[below=0.2cm, scale =0.8] at (x axis mid) {quality of data (SNR)};
\node[left=0.2 cm, scale =0.8, align=left] at (y axis mid) {amount\\ of data};

\draw[ultra thick, color = black] (A) -- (B) to[out= -60,in=178] (D);
\draw[ultra thick, dashed, color = black] (A) -- (B) -- (C);

\draw (3.6,3) node [scale = 1,  align = left] {Easy};
\draw (4.5,1.4) node [scale = 1] {Hard}; 
\draw (1.5,.9) node [scale = 1] {Impossible}; 

\draw (9.6,3.2) node [scale = 0.9, align = left, color = blue!80!black, text width =  7.1cm] {Easy -- Problems can be solved using \\ computationally efficient methods.};
\draw (9.6, 1.85) node [scale = 0.9, align = left, color =  red!80!black, text width = 7.1cm] {Hard -- All known efficient methods fail \\ but brute-force methods can still succeed.};
\draw (9.6,.5) node [scale = 0.9, align = left, color = black, text width = 7.1 cm] {Impossible -- All methods fail regardless \\ of computational complexity.};

\end{tikzpicture}
}
\caption{\label{fig:pd}Example phase diagram for a high-dimensional inference problem such as signal estimation for the standard linear model.
The parameter regions indicate the difficulty of inference of some fixed quality.
 } 
\end{figure}

The analysis of high-dimensional inference  problems focuses on two questions:
\begin{itemize}
\item What is the fundamental limit of inference without  computational constraints?
\item What can be inferred from data using computationally efficient methods?
\end{itemize}
It is becoming increasingly common for the answers to these questions to be framed in terms of \emph{phase diagrams}, which provide important information about fundamental tradeoffs involving the amount and quality of data. For example, the phase diagram in Figure~\ref{fig:pd} shows that increasing the amount of data not only provides more information, but also moves the problem into a regime where efficient methods are optimal. By contrast, increasing the SNR may lead to improvements that can be attained only with significant computational complexity.

\subsection{Three Approaches to Analysis}

For the standard linear model, the qualitative behavior shown in Figure~\ref{fig:pd} is correct and can be made quantitatively precise in the large-system limit. In general, there are many different approaches that can be used to analyze high-dimensional inference problems.
Three popular approaches are described below.

\subsubsection{Information-Theoretic  Analysis} The standard approach taken in information theory is to first obtain precise characterizations of the fundamental limits, without any computational constraints, and then to use these limits to inform the design and analysis of practical methods.  In many cases, the fundamental limits can be understood by studying  macroscopic system properties in the large-system limit, such as the mutual information and the minimum mean-squared error (MMSE). There are a wide variety of mathematical tools to analyze these quantities in the context of  compression and communication~\cite{cover:2006, han:2004}. Unfortunately, these tools alone are often unable to provide simple descriptions for the behavior of high-dimensional statistical inference problems.

\subsubsection{The Replica Method from Statistical Physics}\index{replica method}  An alternative approach for analyzing the fundamental limits is provided by the powerful but heuristic \textit{replica method} from statistical physics~\cite{mezard:2009,zdeborova:2016b}. This method, which was developed originally in the context of disordered magnetic materials known as spin glasses, has been applied  successfully to a wide variety of problems in science and engineering. At a high level, the replica method consists of  a sequence of derivations that provide explicit formulas for the mutual information in the large-system limit. The main limitation, however,  is that the validity of these formulas relies on certain assumptions that are unproven in general. A common progression in the statistical physics literature is that results are first conjectured using the replica method and then proven using very different techniques. For example, formulas for the  Sherrington-Kirkpatrick were conjectured using the replica method in 1980 by Parisi~\cite{parisi:1980}, but were not rigorously proven until 2006 by Talagrand~\cite{talagrand:2006}. 
 
\subsubsection{Analysis of Approximate Inference} 
Significant work has focused on tractable methods for computing summary statistics of the posterior distribution. Variational inference~\cite{mackay:2003,wainwright:2008,pereyra:2016} refers to a large class of methods where the inference problem is recast as an optimization problem. The well-known mean-field variational approach typically refers to minimizing Kullback-Leibler divergence with respect to product distributions.  More generally, however, the variational formulation also encompasses the Bethe and Kikuchi methods for sparsely connected or highly decomposable models as well as the expectation consistent (EC) approximate inference framework of Opper and Winther~\cite{opper:2005}. 
 
A variety of methods can be used to solve or approximately solve the variational optimization problem, including message passing algorithms such as belief propagation~\cite{pearl:1998}, expectation propagation~\cite{minka:2001}, and approximate message passing~\cite{donoho:2009a}. In some cases, the behavior of these algorithms can be characterized precisely via density evolution (coding theory) or state evolution (compressed sensing), which leads to single-letter characterizations of the behavior in the large-system limit.  \index{approximate message passing}

%% file: section_slm.tex
\section{Problem Setup and Characterization}
\label{sec:setup}

\subsection{Standard Linear Model} \index{standard linear model} 

Consider the inference problem implied by an unobserved random vector $\bX \in \mathbb{R}^N$ and an observed random vector $\bY \in \mathbb{R}^M$.
An important special case is the \emph{Gaussian channel},  where $M=N$ and the unknown variables are related to the observations via  
\begin{align}
Y_n = \sqrt{ s} X_n + W_n,  \qquad n = 1, \dots, N, \label{eq:awgn1}   \index{Gaussian channel} 
\end{align}
where $\{W_n\}$ is iid standard Gaussian noise and $s \in [0, \infty)$ parameterizes the signal-to-noise ratio. The Gaussian channel, which is also known as the Gaussian sequence model in the statistics literature~\cite{johnstone:2015}, provides a useful first-order approximation for a wide variety of applications in science and engineering.

The \emph{standard linear model} is an important generalization of the Gaussian channel in which the observations consist of noisy linear measurements: 
\begin{align}
Y_m = \langle \bA_m,  \bX \rangle  + W_m,  \qquad m = 1, \dots, M, \label{eq:linear_model_1} 
\end{align}
where $\langle \cdot, \cdot \rangle$ denotes the standard Euclidean inner product, $\{\bA_m\}$ is a known sequence of $N$-length measurement vectors, and $\{W_m\}$ is iid Gaussian noise. Unlike the Gaussian channel, the number of observations  $M$ may be different from the number of unknown variables $N$. For this reason the measurement indices are denoted by $m$ instead of $n$. In matrix form, the standard linear model can be expressed as 
\begin{align}
 \bY &= \bA \bX + \bW ,
 \label{eq:linear_model} 
\end{align}  
where $\bA\in \mathbb{R}^{M\times N}$ is a known matrix and $\bW\sim \normal(0, \Id_M)$. Inference problems involving the standard linear model include linear regression in statistics, both channel and symbol estimation in wireless communications, and sparse signal recovery in compressed sensing~\cite{foucart:2013,eldar:2012}. In the standard linear model, the matrix $\bA$ induces dependencies between the unknown variables which make the inference problem significantly more difficult. 

Typical inference questions for the standard linear model include the following: 
\begin{itemize}

\item \textbf{Estimation}  of unknown variables: The performance of an estimator $\bY \mapsto \hat{\bX} $ is often measured using its \emph {mean-squared error} (MSE),
\[
\ex{ \| \bX - \hat{\bX}\|^2}.
\]
The optimal MSE, computed by minimizing over all possible estimators, is called the \emph{minimum mean-squared error} (MMSE),
\[
\ex{ \| \bX - \ex{ \bX \mid \bY} \|^2}.
\]
The MMSE is also equivalent to the Bayes risk under squared-error loss. 

\item \textbf{Prediction} of a new observation $Y_\mathrm{new} = \langle \bA_\mathrm{new},  \bX \rangle + W_\mathrm{new}$: Performance of an estimator $(\bY,\bA_\mathrm{new}) \mapsto \hat{Y}_\mathrm{new}$ is often measured using the \emph{prediction mean-square error},
\[
\ex{ (Y_\mathrm{new} - \hat{Y}_\mathrm{new})^2}.
\]

\item \textbf{Detection} of whether the $i$-th entry belongs to a subset $K$ of the real line: For example, $K=\mathbb{R}\setminus\{0\}$ tests whether entries are non-zero.  In practice, one typically defines a test statistic
\[ T(\by) = \ln \frac{p(\bY=\by | X_i \in K)}{p (\bY=\by | X_i \notin K)} \]
and then uses a threshold rule that chooses $X_i \in K$ if $T(\by) \geq \lambda$ and $X_i \notin K$ otherwise.
The performance of this detection rule can be measured using the true positive rate (TPR) and the false positive rate (FPR) given by
\[
\mathsf{TPR} = p(T(\bY)\geq \lambda | X_i \in K) \qquad \mathsf{FPR} = p(T(\bY)\geq \lambda | X_i \notin K).
\]
The receiver operating characteristic (ROC)  curve for this binary decision problem is obtained by plotting the TPR versus the FPR as a parametric function of the threshold $\lambda$. An example is given in Figure~\ref{fig:roc} below. 

\item \textbf{Posterior marginal approximation} of a subset $S$ of unknown variables: The goal is to compute an  approximation $\hat{p}(\bx_S \mid \bY )$ of the marginal distribution of entries in $S$, which can be used to provide summary statistics and measures of uncertainty. In some cases, accurate approximation of the posterior for small subsets is possible even though the full posterior distribution is intractable. 
\end{itemize}

\subsubsection{Analysis of Fundamental Limits} 
To understand the fundamental and practical limits of  inference with the standard linear model, a great deal of work has focused on the setting where: 1) the entries of $\bX$ are drawn iid from a known prior distribution; and 2) the matrix $\bA$ is an $M\times N$ random matrix whose entries $\bA_{ij}$ are drawn iid from $\mathcal{N}(0,1/N)$. Sparsity in $\bX$ can be modeled by using a spike-slab signal distribution (i.e., a mixture of a very narrow distribution with a very wide distribution). Consider a sequence of problems where the number of measurements per signal dimension converges to $\delta$.  In this case, the normalized mutual information and MMSE corresponding to the large-system limit\footnote{Under reasonable conditions, one can show that these limits are well-defined for almost all non-negative values of $\delta$ and that $\cI(\delta)$ is continuous.} are given by
\[
\cI(\delta) \triangleq \lim_{\substack{ M,N \to \infty \\ M/N\to\delta}} \frac{1}{N} I(\bX; \bY \gmid \bA) , \qquad  \cM(\delta)   \triangleq \lim_{\substack{M,N \to \infty \\ M/N\to\delta}} \frac{1}{N} \mmse(\bX \gmid \bY , \bA).
\]

Part of what makes this problem interesting is that the MMSE can have discontinuities, which are referred to as \textit{phase transitions}. The values of $\delta$ at which these discontinuities occur are of significant interest because they correspond to problem settings in which a small change in the number of measurements makes a large difference in the ability to estimate the unknown variables.  In the above limit, the value of $\cM(\delta)$ is undefined at these points.

\subsubsection{Replica-Symmetric Formula} 
Using the heuristic replica method from statistical physics, Guo and Verdu~\cite{guo:2005} derived single-letter formulas for the mutual information and MMSE in the standard linear model with iid variables and an iid Gaussian matrix:
\begin{align}
\cI(\delta)  &=  \min_{s\geq 0}  \underbrace{\left\{ I\left(X ; \sqrt{s} X + W \right) +  \frac{\delta}{2}\left(  \log\left( \frac{\delta}{s} \right) + \frac{s}{\delta} -1  \right) \right\}}_{\cF(s)} \label{eq:I_limit_IID}   \\
 \cM(\delta) &= \mmse\big(X \gmid \sqrt{ s^*(\delta) \, } X+ W\big). \label{eq:M_limit_IID}
\end{align}
In these expressions,  $X$ is a univariate random variable drawn according to the prior $p_X$, $W \sim \normal(0,1)$ is independent Gaussian noise, 
and  $s^*(\delta)$ is a minimizer of the objective function $\cF(s)$. Precise definitions of mutual information and MMSE are provided in Section~\ref{sec:role_of_MI} below. By construction, the replica mutual information~\eqref{eq:I_limit_IID} is a continuous function of the measurement rate $\delta$. However, the replica MMSE prediction~\eqref{eq:M_limit_IID} may have discontinuities when the global minimizer $s^*(\delta)$ jumps from one minimum to another and $\cM(\delta)$ is well-defined only if $s^*(\delta)$ is the unique minimizer. In~\cite{reeves:2016a,reeves:2016,reeves:2016d}, the authors prove these expressions are exact for the standard linear model with an iid Gaussian measurement matrix.
An overview of this proof is presented in Section~\ref{sec:proof}.

\subsubsection{Approximate Message Passing} An algorithmic breakthrough for the standard linear model with iid Gaussian matrices was provided by the approximate message passing (AMP) algorithm \cite{donoho:2009a,bayati:2011} and its generalizations \cite{rangan:2011,vila:2013,ma:2014,metzler:2016,schniter:2016a,rangan:2017a}.  For CDMA waveforms, the same idea was applied earlier in~\cite{kabashima:2003a}.

An important property of this class of algorithms is that the  performance for large iid Gaussian matrices is characterized precisely via a state evolution formalism~\cite{bayati:2011,bayati:2012a}. Remarkably, the fixed-points of the state evolution correspond to the stationary points of the objective function $\cF(s)$ in~\eqref{eq:I_limit_IID}. For cases where the replica formulas are exact, this means that AMP-type algorithms can be optimal with respect to marginal inference problems whenever the largest local minimizer of $\cF(s)$ is also the global minimizer~\cite{reeves:2017e}.

\subsubsection{The Generalized Linear Model}
In the generalized linear model, the observations $\bY \in \mathbb{R}^M$ are related to the unknown variables $\bX \in \mathbb{R}^N$ by way of the $M \times N$ matrix $\bA$, the random variable $\bZ = \bA \bX$, and the conditional distribution
\begin{align}
p_{\bY \mid \bX} (\by \mid \bx) \triangleq p_{\bY \mid \bZ}(\by \mid \bA \bX) , \label{eq:GML}
\end{align}
where $p_{\bY \mid \bZ}(\by \gmid \bz)$ defines a memoryless (i.e., separable) channel.
The generalized linear model is fundamental to generalized linear regression in statistics. It is also used to model different sensing architectures (e.g., Poisson channels, phase retrieval) and the effects of scalar quantization.
The AMP algorithm was introduced by Donoho, Maleki, and Montanari in~\cite{donoho:2009a} and extended to the GLM by Rangan in~\cite{rangan:2011}.
More recent work has focused on AMP-style algorithms for rotationally invariant random matrices~\cite{cakmak:2014,fletcher:2016,rangan:2016a,schniter:2016b,cakmak:2017a,he:2016}.

\subsection{Illustrative Examples} 

\begin{figure}
\centering
\scalebox{1.3}{%
\begin{tikzpicture}
    	\node[anchor=south west,inner sep=0, scale = 1.1] at (0,0) {\includegraphics[height= 1.8in]{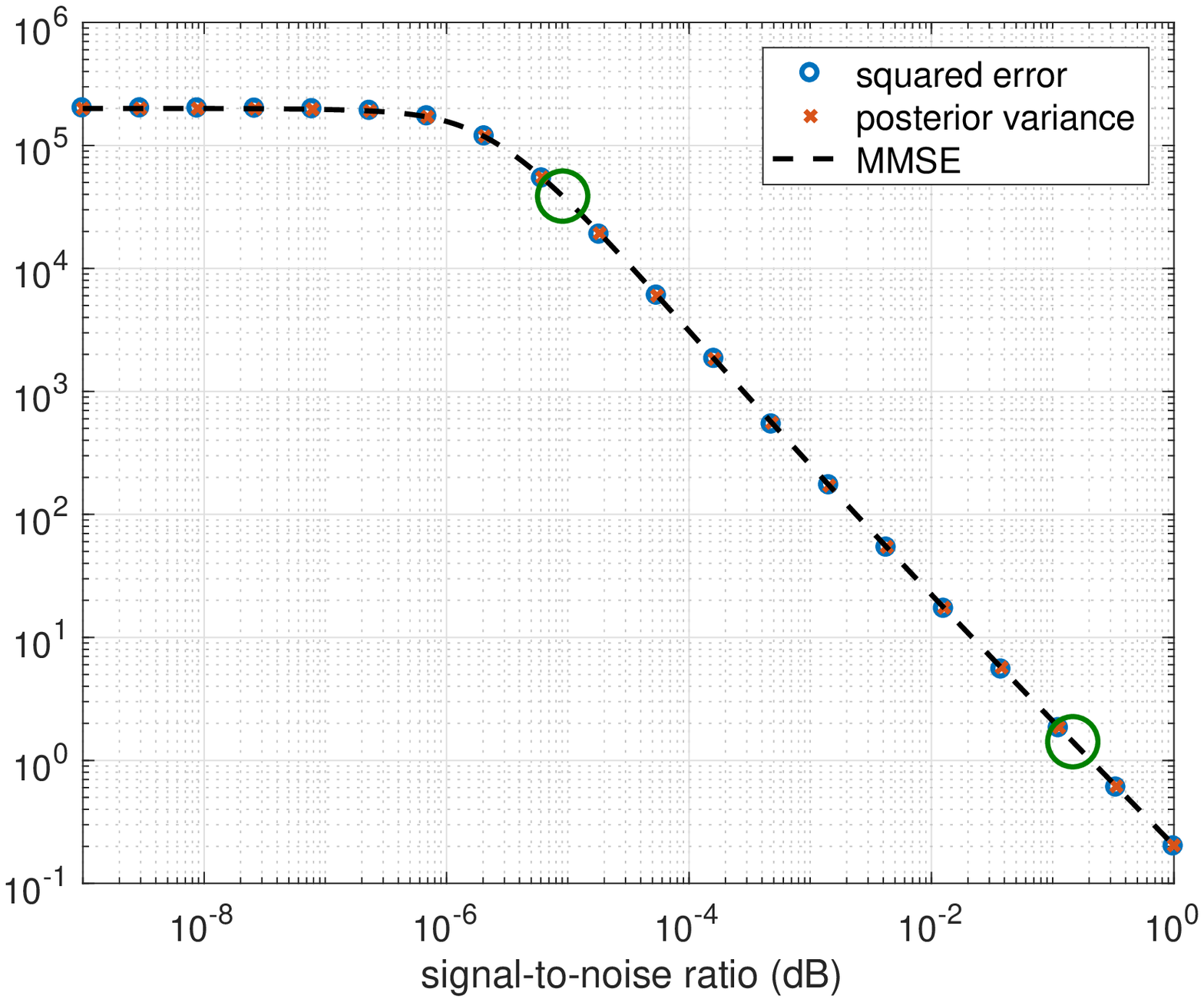}
\includegraphics[height= 1.8 in]{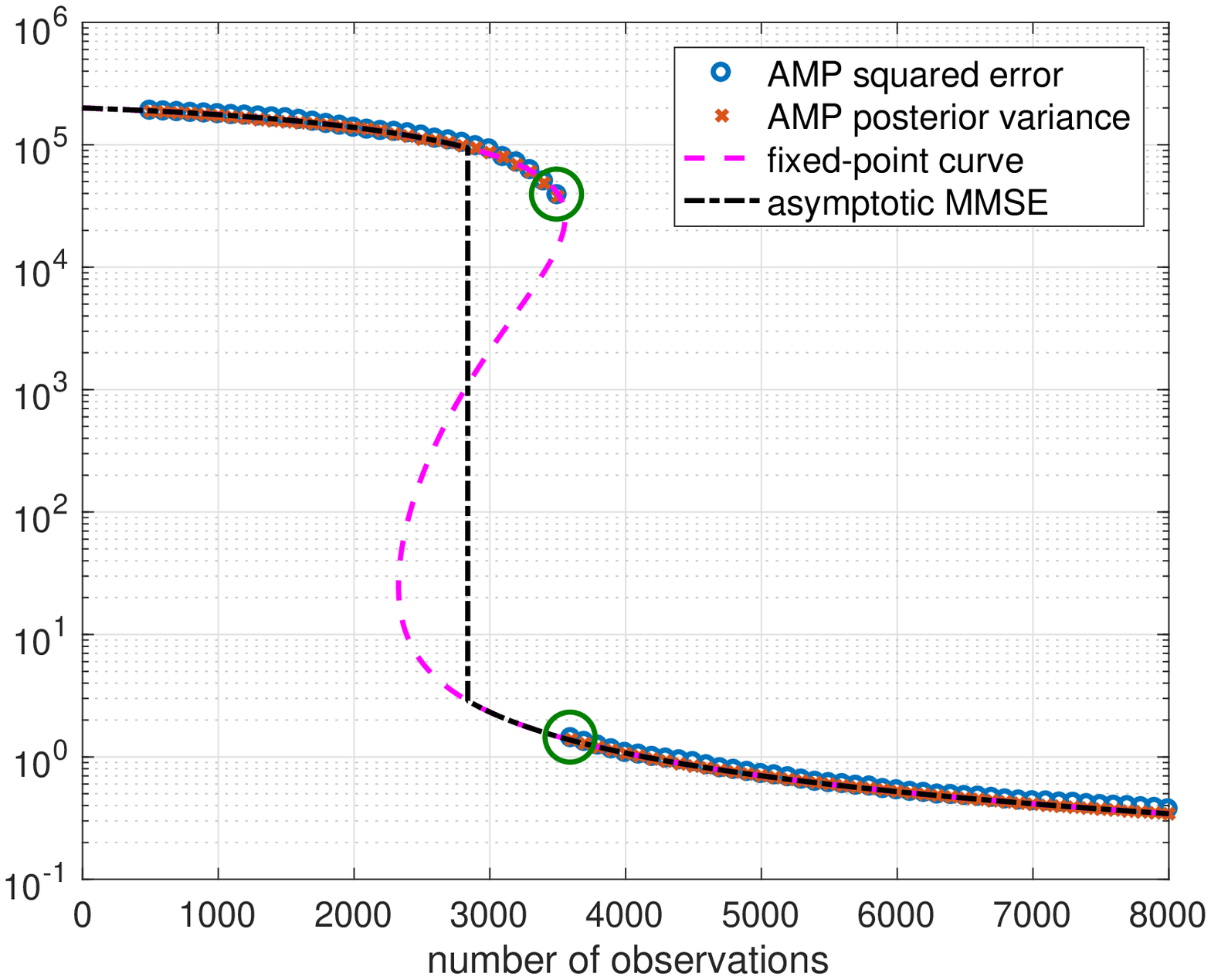}
};
	 \node[ scale = 0.8 ] at (3.1,4.3) {A};
	  \node[ scale = 0.8 ] at (5.6,1.5) {B};
	   \node[ scale = 0.8 ] at (9.3,4.3) {A};
	  \node[ scale = 0.8 ] at (9.3,1.5) {B};
	  
	  \node[ ] at (3.1,5.3) {Gaussian Channel};
	\node[ ] at (9.3,5.3) {Standard Linear Model};
	 \end{tikzpicture}}

\caption{\label{fig:mmse_amp_comp}%
Comparison of average squared error for the Gaussian channel as a function of signal-to-noise ratio  (left panel) and the standard linear model as a function of the number of observations (right panel). In both cases, the unknown variables are iid Bernoulli-Gaussian with zero mean and a fraction $\gamma = 0.20$ of nonzero entries.
}

\end{figure}

We now consider some examples that illustrate similarities and the differences between the Gaussian channel and the standard linear model. In these examples, the unknown variables are drawn iid from the Bernoulli-Gaussian distribution, which corresponds to the product of independent Bernoulli and Gaussian random variables and is given by
\begin{align}
\BG(x \mid \mu, \sigma^2, \gamma ) & \triangleq  (1-\gamma)  \delta_0(x )  + \gamma\,  \normal(x \mid \mu, \sigma^2). \label{eq:BG}
\end{align}
Here,  $\delta_0$ denotes a Dirac distribution with all probability mass at zero and $\mathcal{N}(x|\mu,\sigma^2)$ denotes the Gaussian pdf $(2\pi \sigma^2)^{-1/2} e^{-(x-\mu)^2/(2\sigma^2)}$ with mean $\mu$ and variance $\sigma^2$. The parameter $\gamma \in (0,1)$  determines the expected fraction of non-zero entries. The mean and variance of a random variable $X \sim \BG(x \mid \mu, \sigma^2, \gamma ) $ are given by
\begin{align}
\ex{ X} & = \gamma \mu, \qquad \var(X)  = \gamma(1 -\gamma) \mu^2  + \gamma \sigma^2 .\label{eq:BG_mv}
\end{align}

\begin{figure}
\centering
\scalebox{1.2}{%
\begin{tikzpicture}
    	\node[anchor=south west,inner sep=0, scale = 0.4] at (0,0) {\includegraphics
	{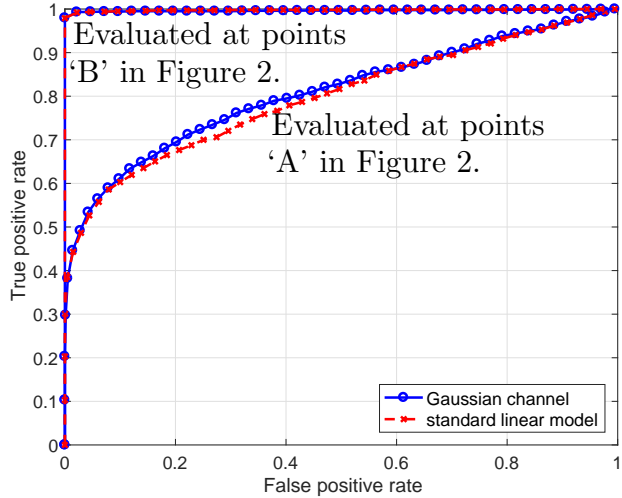}};
\node[align =left, scale = 1, text width= 5 cm ] at (3.2,4.9) {Evaluated at points \\ `B' in Figure~\ref{fig:mmse_amp_comp}. };
	  \node[align =left, scale = 1, text width= 5 cm ] at (5.4,3.9) {Evaluated at points\\ `A' in Figure~\ref{fig:mmse_amp_comp}. };
\end{tikzpicture}}

\caption{\label{fig:roc}ROC curves for detecting the nonzero variables in the parameter regimes labeled `A' and `B' in Figure~\ref{fig:mmse_amp_comp}. The curves for the Gaussian channel (blue) are obtained by thresholding the true posterior inclusion probabilities $\{\gamma_n\}$. The curves for the standard linear model (red) are obtain by thresholding the AMP approximations of the posterior inclusion probabilities.}
\end{figure}

\subsubsection{Gaussian Channel with Bernoulli-Gaussian Prior} 

If the unknown variables are iid $\BG(x \mid \mu, \sigma^2, \gamma ) $ and  the observations are generated according to the Gaussian channel~\eqref{eq:awgn1} then the posterior distribution \emph{decouples} into the product of its marginals: 
\begin{align}
p\left(\bx \mid \by \right) & = \prod_{n=1}^N p\left(x_n \mid y_n \right). \label{eq:pxy_decomple} 
\end{align}
Furthermore, the posterior marginal $p(x_n \mid y_n)$ is also a Bernoulli-Gaussian distribution but with new parameters $(\mu_n, \sigma^2_n, \gamma_n)$ that depend on $y_n$:
\begin{align}
\mu_n &=
\mu + \frac{\sqrt{s} \sigma^2}{1 + s \sigma^2} (y_n -  \sqrt{s} \mu) \\
\sigma^2_n & =   \frac{\sigma^2}{1 + s \sigma^2} \\
\gamma_n &=
\left[ 1 +  \frac{ (1- \gamma) \sqrt{1 + s \sigma^2}}{\gamma}  \exp\left(  \frac{s\mu^2  - 2 \sqrt{s} \mu y_n   - s \sigma^2 y_n^2    }{ 2(1+s\sigma^2)}  \right) \right]^{-1}.
\end{align}
Given these parameters, the posterior mean $\ex{X_n \mid Y_n}$  and posterior variance  $\var(X_n \mid Y_n)$ can be computed using~\eqref{eq:BG_mv}. The parameter $\gamma_n$ is the conditional probability that $X_n$ is nonzero given $Y_n$. This parameter is often called the \emph{posterior inclusion probability} in the statistics literature.

The decoupling of the posterior distribution makes it easy to characterize the fundamental limits of performance measures. For example the MMSE is the expectation of the posterior variance $\var(X_n \mid Y_n)$ and the optimal tradeoff between the true positive rate and false positive rate for detecting the event $\{ X_n \ne 0\}$ is characterized by the distribution of $\gamma_n$. 

To investigate the statistical properties of the posterior distribution we perform a numerical experiment. First, we draw $N = 10,000$ variables according to the Bernoulli-Gaussian variables with $\mu = 0$, $\sigma^2 = 10^6$, and prior inclusion probability $\gamma = 0.2$. Then, for various values of the signal-to-noise ratio parameter $s$, we evaluate the posterior distribution corresponding to the output of the Gaussian channel. 

In Figure~\ref{fig:mmse_amp_comp} (left panel), we plot three quantities associated with the estimation error: 
\begin{align*}
\text{average squared error:} \quad &\frac{1}{N}  \sum_{n=1}^N\left (X_n - \ex{ X \mid Y_n} \right)^2\\
\text{average posterior variance:} \quad & \frac{1}{N} \sum_{n=1}^N \var(X_n \mid Y_n)\\
\text{average MMSE:} \quad & \frac{1}{N} \sum_{n=1}^N \ex{ \var(X_n \mid Y_n)}.
\end{align*}
Note that the squared error and posterior variance are both random quantities because they are functions of the data. This means that the corresponding plots would look slightly different if the experiment were repeated multiple times. The MMSE,  however, is a function of the joint distribution of $(\bX, \bY)$ and is thus non-random. In this setting, the fact that there is little difference between the \emph{averages} of these quantities can be seen as a consequence of the decoupling of the posterior distribution and the law of large numbers.

In Figure~\ref{fig:roc}, we plot the ROC curve for the problem of detecting the nonzero variables. The curves are obtained by thresholding the posterior inclusion probabilities $\{ \gamma_n\}$ associated with the values of the signal-to-noise ratio at the points A and B in the figure.

\subsubsection{Standard Linear Model with Bernoulli-Gaussian Prior} 
Next, we consider the setting where 
the observations are generated  by the standard linear model~\eqref{eq:linear_model}. In this case, the measurement matrix introduces dependence in the posterior distribution and the decoupling seen in \eqref{eq:pxy_decomple} does not hold in general. 

To characterize the posterior distribution, one can use that fact that $\bX$ is conditionally Gaussian given the support vector $\bU \in \{0,1\}^N$, where $X_n=0$ if $U_n = 0$ and $X_n \neq 0$ with probability one if $U_n = 1$.  Consequently, the posterior distribution can be expressed as a Gaussian mixture model of the form
\begin{align*}
p(\bx \mid \by, \bA)  =  \sum_{\bu \in \{0,1\}^N} p(\bu \mid  \by, \bA)\, \normal\left(\bx \mid \ex{ \bX \mid \by, \bA, \bu} , \cov(\bX \mid \by, \bA, \bu) \right) ,
\end{align*}
where the summation is over all possible support sets. The posterior probability of the support set is given by 
\begin{align*}
p(\bu \mid \by, \bA) \propto p(\bu) \normal\left(\by \mid\mu  \bA_\bu \one   ,  \Id + \sigma^2 \bA_{\bu}^T \bA_{\bu} \right), 
\end{align*}
where $\bA_{\bu}$ is the submatrix of $\bA$ formed by removing columns where $u_n = 0$ and $\one$ denotes a vector of ones. The posterior marginal is obtained by integrating out the other variables to get 
\begin{align*}
p(x_n \mid \by, \bA)  =  \int p(\bx \mid \by, \bA)  \, \dd \bx_{\sim n},
\end{align*}
where $\bx_{\sim n}$ denotes all the entries except for $x_n$. 

Here, the challenge is that the number of terms in the summation over $\bu$ grows exponentially with the signal dimension $N$. Since it is difficult to compute the posterior distribution in general, we use AMP \index{approximate message passing|(} to compute approximations to the posterior marginal distributions. The marginals of the approximation, which belong to the Bernoulli-Gaussian family of distributions, are given by
\begin{align}
\hat{p}(x_n \mid \by, \bA) = \BG(x_n \mid  \mu_n, \sigma^2_n, \gamma_n), \label{eq:phat}
\end{align}
where the parameters $(\mu_n, \sigma^2_n, \gamma_n)$ are the outputs of the AMP algorithm.

Similar to the previous example, we perform a numerical experiment to investigate the statistical properties of the marginal approximations. First, we draw $N = 10,000$ variables according to the Bernoulli-Gaussian variables with $\mu=0$, $\sigma^2 = 10^6$,  and prior inclusion probability $\gamma = 0.2$. Then, for various values of $M$, we obtain measurements from the standard linear model with iid Gaussian measurement vectors $\bA_m \sim \normal(0,N^{-1} \Id)$ and use AMP to compute the parameters $(\mu_n, \sigma^2_n, \gamma_n)$ used in the marginal posterior approximations.

In Figure~\ref{fig:mmse_amp_comp} (right panel), we plot the squared error and the approximation of the posterior variance associated with the AMP marginal approximations: 
 \begin{align*}
 \text{average AMP  squared error:} \quad &\frac{1}{N}  \sum_{n=1}^N\left (X_n - \bEx_{\hat{p}}\left[ X_n \mid \bY, \bA\right] \right)^2\\
\text{average AMP posterior variance:} \quad & \frac{1}{N} \sum_{n=1}^N \var_{\hat{p}}(X_n \mid \bY, \bA).
\end{align*}
In these expressions, the expectation and the variance are computed with respect to the marginal approximation in~\eqref{eq:phat}. Because these quantities are  functions of the random data, one expects that they would look slightly different if the experiment were repeated multiple times.

At this point, there are already some interesting observations that can be made. First, we note that the AMP approximation of the mean can be viewed as a point-estimate of the unknown variables. Similarly,  the AMP approximation of the posterior variance (which depends on the observations but not the ground truth) can be viewed as a point-estimate of the squared error. From this perspective, the close correspondence between the squared error  and the AMP approximation of the variance seen in Figure~\ref{fig:mmse_amp_comp} suggests that AMP is self-consistent in the sense that it provides an accurate estimate of its square error.

Another observation is that the squared error undergoes an abrupt change at around 3,500 observations, between the points labeled `A' and 'B'. Before this point, the squared error is within an order of magnitude of the prior variance. After this point, the squared error drops discontinuously. This illustrates that the estimator provided by AMP is quite accurate in this setting.

However, there are still some important questions that remain. For example, how accurate are the AMP posterior marginal approximations? Is it possible that a different algorithm (e.g., one that computes the true posterior marginals) would lead to estimates with significantly smaller squared error? Further questions concern how much information is lost in focusing only on the marginals of the posterior distribution as opposed to the full posterior distribution.

Unlike the Gaussian channel, it is not possible to evaluate the MMSE directly because the summation over all $2^{10000}$ support vectors is prohibitively large. For comparison, we plot the large-system MMSE predicted by \eqref{eq:M_limit_IID}, which corresponds to the large-$N$ limit where the fraction of observations is parametrized by $\delta = M/N$. t 

The behavior of the large-system MMSE is qualitatively similar to the AMP squared error because it has a single jump discontinuity (or phase transition). However, the jump occurs after only 2,850 observations for the MMSE as opposed to after 3,600 observations for AMP. By comparing the AMP squared error with the asymptotic MMSE, we see that the AMP marginal approximations are accurate in some cases (e.g., when the number of observations is less than 2,840 or greater than 3,600) but highly inaccurate in others (e.g., when the number of observations is between 2,840 and 3,600).

In Figure~\ref{fig:roc}, we plot the ROC curve for the problem of detecting the nonzero variables in the standard linear model. In this case, the curves are obtained by thresholding AMP approximations of the posterior inclusion probabilities $\{ \gamma_n\}$.  It is interesting to note that the ROC curves corresponding to the two different observation models (the Gaussian channel and the standard linear model) have similar shapes when they are evaluated in problem settings with matched squared error. 
\index{approximate message passing|)} 

%% file: section_role_of_MI.tex

\section{The Role of Mutual Information and MMSE} 
\label{sec:role_of_MI}

The amount one learns about an unknown vector $\bX$ from an observation $\bY$ can be quantified in terms of the difference between the prior and posterior distributions. For a particular realization of the observations $\by$, a fundamental measure of this difference is provided by the relative entropy
\begin{align*}
\DKL{ p_{\bX \mid \bY}(\cdot \mid \by) }{p_{\bX}(\cdot)}.
\end{align*}  
This quantity is nonnegative and equal to zero if and only if the posterior is the same as the prior almost everywhere.

For real vectors, another way to assess the difference between the prior and posterior distribution is to compare the first and second moments of their distributions. These moments are summarized by the mean $\ex{\bX}$, the conditional mean $\ex {\bX \mid \bY =\by}$, the covariance matrix 
\begin{align*}
\cov(\bX) & \triangleq \ex{ \left( \bX - \ex{ \bX}\right) \left( \bX - \ex{ \bX}\right)^T},
\end{align*}
and the conditional covariance matrix
\begin{align*}
\cov(\bX\mid \bY = \by ) & \triangleq \ex{ \left( \bX - \ex{ \bX \mid \bY}\right) \left( \bX - \ex{ \bX \mid \bY }\right)^T \mid \bY = \by}.
\end{align*}
Together, these provide some measure of how much ``information'' is in the data.

One of the difficulties of working with the posterior distribution directly is that it can depend non-trivially on the particular realization of the data. It can be much easier to focus on the behavior for \emph{typical} realizations of the data by studying the distribution of the relative entropy when $\bY$ is drawn according to the marginal distribution  $p(\by)$. For example, the expectation of the relative entropy is the mutual information
\begin{align*}
 I(\bX ; \bY)  & \triangleq \bEx \left[ \DKL{ p_{\bX \mid \bY}(\cdot \mid \bY) }{p_{\bX}(\cdot)} \right] .
\end{align*}
Similarly, the expected value of the conditional covariance matrix\footnote{Observe that $\cov(\bX\mid \bY )$ is a random variable in the same sense as $\ex{\bX \mid \bY}$.} is
\begin{align*}
\ex{ \cov(\bX\mid \bY )} & = \ex{ \left( \bX - \ex{ \bX \mid \bY}\right) \left( \bX - \ex{ \bX \mid \bY }\right)^T }.
\end{align*}
The trace of this matrix equals the Bayes risk for squared-error loss, which is more commonly called the MMSE and defined by
\begin{align*}
\mmse(\bX \mid \bY)  \triangleq  \gtr( \ex{ \cov(\bX\mid \bY )}) = \ex{ \left\| \bX - \ex{ \bX \mid \bY} \right \|_2^2},
\end{align*}
where $\|\cdot\|_2$ denotes the Euclidean norm.
Part of the appeal of working with the mutual information and the MMSE is that they satisfy a number of useful functional properties, including chain rules and data-processing inequalities.

The prudence of focusing on the expectation with respect to the data depends on the extent to which the random quantities of interest deviate from their expectations. In the statistical physics literature, the concentration of the relative entropy and squared error around the their expectations is called the \emph{self-averaging property} and is often assumed for large systems~\cite{mezard:2009}.

\subsection{I-MMSE relationships for the Gaussian channel}\label{sec:I_MMSE} 

Given an $N$-dimensional random vector $\bX = (X_1, \dots , X_N)$, the output of the Gaussian channel with signal-to-noise ratio parameter $s \in [0, \infty)$ is denoted by
\begin{align*}
\bY(s) =  \sqrt{s} \, \bX + \bW,
\end{align*}
where $\bW \sim \normal(0, \Id_N)$ is independent Gaussian noise. Two important functionals of the joint distribution of $(\bX, \bY_s)$  are the \emph{mutual information function}: 
\begin{align*}
I_{\bX}(s) &  =\frac{1}{N}  I(\bX; \bY(s))
\end{align*}
and the \emph{MMSE function}:
\begin{align*}
M_{\bX}(s) &  = \frac{1}{N}   \ex{ \left\| \bX - \ex{\bX \mid  \bY(s)} \right\|_2^2}. 
\end{align*}
 In some cases, these functions can be 
computed efficiently using numerical integration or Monte Carlo approximation. For example, if the entries of  $\bX$ are iid copies of a scalar random variable $X$ then the mutual information and MMSE depend only on the marginal distribution:
\begin{align*}
 I_{\bX}(s) &  =I(X; \sqrt{s} X + W)\\
 M_{\bX}(s) &  =  \ex{\left(  X - \ex{X \mid  \sqrt{s} X + W} \right)^2} .
\end{align*}
Another example is if $\bX$ is drawn according to a Gaussian mixture model with a small number of mixture components. For general high-dimensional distributions, however, direct computation of these functions can be intractable due to the curse of dimensionality. Instead, one often resorts to asymptotic approximations.

The I-MMSE relationship \cite{guo:2005a} asserts that the derivative of the mutual information is one half of the MMSE,
\begin{align}
\frac{\dd}{ \dd s} I_{\bX}(s) = \frac{1}{2} M_{\bX}(s) . \index{I-MMSE relationship}
\end{align}
This result is equivalent to the classical  De Bruijn identity~\cite{stam:1959}, which relates the derivative of differential entropy to the Fisher information. Part of the significance of the I-MMSE relationship is that it provides a link between an information-theoretic quantity and an estimation-theoretic quantity.

Another important property of the MMSE function is that its derivative is
\begin{align}
M'_{\bX}(s) &  = - \frac{1}{N}  \ex{\| \cov(\bX \mid \sqrt{s} \, \bX + \bW)\|_F^2}, \label{eq:Mprime}
\end{align}
where $\|\cdot\|_F$ denotes the Frobenious norm.  Since the derivative is non-positive, it follows that the MMSE function is non-increasing and the mutual information function is concave.

The relationship between the MMSE function and its derivative imposes some useful constraints on the MMSE. One example is the so-called \emph{single-crossing property} \cite{guo:2011}, which asserts that for any random vector $\bX$ and isotropic Gaussian random vector $\bZ$, the MMSE functions $M_{\bX}(s)$ and $M_{\bZ}(s)$ cross at most once. 

The following result states a monotonicity property concerning a transformation of the MMSE function. A matrix generalization of this result is given in \cite{reeves:2018a}. 

\begin{theorem}[Monotonicity of MMSE]\label{thm:monotonicity_mmse}
For any random vector $\bX$ that is not almost-surely constant, the function
\begin{align}
k_{\bX}(s) \triangleq \frac{1}{M_{\bX}(s)} -s
\end{align}
is well-defined and non-decreasing on $(0, \infty)$.
\end{theorem} 
\begin{proof}
The MMSE function is real analytic, and hence infinitely differentiable,  on $(0,\infty)$ \cite{guo:2011}. By differentiation,  one finds that 
\begin{align}
\frac{\dd}{ \dd s} k_{\bX}(s) & = -  M'_{\bX}(s) /  M^2_{\bX}(s)- 1.  \label{eq:Kprime}
\end{align}
Let $\lambda_1 , \dots , \lambda_N \in [0,\infty)$ be the eigenvalues of the $N \times N$ matrix $\ex{ \cov(\bX \mid \bY))}$ where $\bY = \sqrt{s} \, \bX+ \bW$. Starting with \eqref{eq:Mprime}, we find that
\begin{align*}
 - M'_{\bX}(s)  & =    \frac{1}{N}  \ex{ \left \| \cov(\bX \mid \bY)  \right \|_F^2}  \ge  \frac{1}{N}  \left \|  \ex{\cov(\bX \mid \bY)}  \right \|_F^2\\
 &   =  \frac{1}{N} \sum_{n=1}^N \lambda^2_n  \ge  \left(  \frac{1}{N} \sum_{n=1}^N \lambda_n  \right)^2   =  M_{\bX}^2(s),
\end{align*}
where both inequalities are due to Jensen's inequality. Combining this inequality with \eqref{eq:Kprime} establishes that the derivative of $k_{\bX}(s)$ is non-negative and hence $k_{\bX}(s)$ is non-decreasing.
\end{proof}

We remark Theorem~\ref{thm:monotonicity_mmse} implies the single-crossing property. To see this, note that if $\bZ \sim \normal(0, \sigma^2 I)$ then $k_{\bZ}(s) = \sigma^{-2}$ is a constant and thus $k_{\bX}(s)$ and $k_{\bZ}(s)$ cross at most once.  Furthermore, Theorem~\ref{thm:monotonicity_mmse} shows that for many problems, the Gaussian distribution plays an extremal role for distributions with finite second moments. For example, if we let $\bZ$ be a Gaussian random vector with the same mean and covariance as $\bX$, then we have
\begin{align*}
I_{\bX}(s) &\le  I_{\bZ}(s)\\ 
M_{\bX}(s) &\le M_{\bZ}(s)\\ 
M_{\bX}' & \le M_{\bZ}'(s)
\end{align*} 
where equality holds if and only if $\bX$ is Gaussian.  The importance of these inequalities follows from the fact that the Gaussian distribution is easy to analyze and often well behaved.

\subsection{Analysis of Good Codes for the Gaussian Channel}
\label{sec:good_codes}

This section provides an example of the how the properties described in Section~\ref{sec:I_MMSE} can be applied in the context of a high-dimensional inference problem. The focus is on the channel coding problem for the Gaussian channel. A code for the Gaussian channel is a collection $\cX = \{ \bx(1), \dots \bx(L) \}$ of $L$ codewords in $\mathbb{R}^N$ such that $\bX = \bx(J)$ where $J$ is drawn uniformly from $\{1,2,\ldots,L\}$. 
The output of the channel is given by 
\begin{align}
\bY = \sqrt{ \snr} \, \bX + \bW,  \label{eq:awgn} 
\end{align}
where $\bW \sim \normal(0, \Id)$ is independent Gaussian noise.  

The code is called $\eta$-good if it satisfies three conditions: 
\begin{itemize}
\item \emph{Power Constraint}: The codewords satisfy the average power constraint 
\begin{align}
\frac{1}{N} \ex{ \|\bX\|_2^2}  \le 1. \label{eq:power} 
\end{align}

\item \emph{Low Error Probability:} For the MAP decoding decision $\hat{\bX}$, we have
\[ \Pr \left( \hat{\bX} = \bX \right) \geq \pr{ p_{\bX|\bY} (\bx(J) | \bY ) > \max_{\ell\neq J} p_{\bX|\bY} (\bx(\ell) | \bY ) } \geq 1-\eta. \]
\item \emph{Sufficient Rate:} The number of codewords satisfies $L \geq  (1 + \snr)^{(1- \eta)N/2}.$
\end{itemize}

\begin{lemma}[Corollary of the Channel Coding Theorem]
For every $\snr > 0$ and $\eps >0$ there exists an integer $N$ and random vector $\bX = (X_1, \dots X_N)$ satisfying the average power constraint \eqref{eq:power}  as well as the following inequalities: 
\begin{align}
I_{\bX}(\snr) &\ge \frac{1}{2} \log(1 + \snr) - \eps  \label{eq:I_awgn} \\
M_{\bX}(\snr) & \le \eps  \label{eq:M_awgn}. 
\end{align}
\end{lemma}

The distribution on $\bX$ induced by a good code is fundamentally different from the iid Gaussian distribution that maximizes the mutual information. For example, a good code  defines a discrete distribution that has finite entropy whereas the Gaussian distribution is continuous and hence has infinite entropy. Furthermore,  while the MMSE of a good code can be made arbitrarily small (in the large $N$ limit), the MMSE of the Gaussian channel is  lower bounded by $1/(1 + \snr)$ for all $N$. 

Nevertheless, the distribution induced by a good code and the Gaussian distribution are similar in the sense that their  mutual-information functions must become arbitrarily close for large $N$. It is natural  to ask whether this closeness implies other similarities between the good code and the Gaussian distribution.   The next result shows that closeness in mutual information also implies closeness in MMSE.

\begin{theorem}\label{thm:MMSE_good_code}
For any $N$-dimensional random vector $\bX$ satisfying  the average power constraint \eqref{eq:power} and the mutual information lower bound \eqref{eq:I_awgn} the MMSE function satisfies
\begin{align}
 \frac{e^{-2\eps}}{1+ s }  - \frac{1 - e^{-2\eps}}{ \snr  - s}  \le M_{\bX}(s) & \le \frac{1}{1 + s} \label{eq:M_awgn_bnd}
\end{align}
for all $0 \le s < \snr$. 
\end{theorem}
\begin{proof}
For the upper bound, we have
\begin{align}
M_{\bX}(s) & = \frac{1}{k_{\bX}(s) + s} \le  \frac{1}{k_{\bX}(0) + s}   \le \frac{1}{1 + s} 
\end{align}
where the first inequality follows from Theorem~\ref{thm:monotonicity_mmse}  and the second inequality holds because the  assumption $\frac{1}{N} \ex{ \|\bX\|_2^2} \le 1$ implies that $M_{\bX}(0)  \le 1$, and hence $k_{\bX}(0) \ge 1$.  For the lower bound, we use the following chain of inequalities: 
\begin{align}
\log\left( \frac{(1 + \snr)( k_{\bX}(s) + s)}{ (1 + s)(k_{\bX}(s)  + \snr)} \right)
& = \int_s^{\snr}  \left(  \frac{1}{1 + t}  - \frac{1}{k_{\bX}(s) + t}   \right) \, \dd t  \\ 
& \le \int_s^{\snr}  \left(  \frac{1}{1 + t}  - \frac{1}{k_{\bX}(t) + t}  \right)\, \dd t \label{eq:M_awgn_bnd_aa} \\
& = \int_s^{\snr}   \left( \frac{1}{1 + t}  - M_{\bX}(t) \right)\, \dd t  \\
& \le \int_0^{\snr} \left(   \frac{1}{1 + t}  - M_{\bX}(t) \right) \, \dd t \label{eq:M_awgn_bnd_a}\\
& = \log(1 + \snr) - 2 I_{\bX}(\snr) \label{eq:M_awgn_bnd_b}\\
& \le 2 \eps, \label{eq:M_awgn_bnd_c}
\end{align} 
where \eqref{eq:M_awgn_bnd_aa} follows from Theorem~\ref{thm:monotonicity_mmse},  \eqref{eq:M_awgn_bnd_a} holds because of the upper bound in \eqref{eq:M_awgn_bnd} ensures that the integrand is non-negative, and \eqref{eq:M_awgn_bnd_c} follows from the assumed lower bound on the mutual information. Exponentiating both sides, rearranging terms, and recalling the definition of $k_{\bX}(s)$ leads to the stated lower bound. 
\end{proof}

An immediate consequence of Theorem~\ref{thm:MMSE_good_code} is that the MMSE function associated with a sequence of good codes undergoes a phase transition in the large-$N$ limit. In particular,
\begin{align}
\lim_{N \to \infty} M_{\bX}(s) & = 
\begin{dcases}
\frac{1}{1 + s} , & 0 \le s < \snr\\
0, & \snr < s.
\end{dcases} \label{eq:good_code_pt}
\end{align}
The case $s\in [  \snr, \infty)$ follows from the definition of a good code and the monotonicity of the MMSE function. The case $s \in [0,  \snr)$ follows from Theorem~\ref{thm:MMSE_good_code} and the fact that $\eps$ can be arbitrarily small. 

An analogous result for binary linear codes on the Gaussian channel can be found in~\cite{bhattad:2007}.
The characterization of the asymptotic MMSE for good Gaussian codes, described by~\eqref{eq:good_code_pt}, was also obtained previously using ideas from statistical physics \cite{merhav:2010}. The  derivation presented in this chapter, which relies only on the monotonicity of $k_{\bX}(s)$, bypasses some technical difficulties encountered in the previous approach.

\subsection{Incremental-Information Sequence}\label{sec:inc_mi}  
We now consider a different approach for decomposing the mutual information between random vectors $\bX = (X_1, \dots X_N)$ and $\bY = (Y_1, \dots Y_M)$. The main idea is to study the increase in information associated with new observations.  In order to make general statements, we average over all possible presentation orders for the elements of $\bY$.   To this end, we define the \emph{information sequence}  $\{ I_{m} \}$ according to: 
\begin{align*}
I_m & \triangleq \frac{1}{M!} \sum_{\pi }  I(\bX ; Y_{\pi(1)}, \dots Y_{\pi(m)}), \qquad m = 1, \dots, M,
\end{align*}
where the sum is over all permutations $\pi :[ M]\to [M]$. Note that each summand on the right-hand side is the mutual information between $\bX$ and an $m$-tuple of the observations. The average over all possible permutations can be viewed as the expected mutual information when the order of observations is chosen uniformly at random. 

Due to the random ordering, we will find that $I_m$ is an increasing sequence with: $0  = I_0 \le I_1 \le \dots \le I_{M-1}  \le I_M = I(\bX; \bY)$. To  study the increase in information with additional observations,  we focus on the first- and second-order difference sequences, which are defined as follows: 
\begin{align*}
I'_{m}& \triangleq  I_{m+1}- I_m \\
I''_{m}& \triangleq  I'_{m+1}- I'_m.
\end{align*}
Using the chain rule for mutual information, it is straightforward to show that the first and second order difference sequences can also be expressed as  
\begin{align}
I'_{m} & =   \frac{1}{M!}  \sum_{\pi }  I(\bX ; Y_{\pi(m+1)} \mid Y_{\pi(1)}, \dots Y_{\pi(m)}) \notag \\
I''_{m} & =     \frac{1}{M!}  \sum_{\pi }  I(Y_{\pi(m+2)}; Y_{\pi(m+1)} \mid \bX, Y_{\pi(1)}, \dots Y_{\pi(m)})  \notag\\
& \quad  -    \frac{1}{M!}  \sum_{\pi }  I(Y_{\pi(m+2)}; Y_{\pi(m+1)} \mid Y_{\pi(1)}, \dots Y_{\pi(m)}) . \label{eq:Ipp_alt} 
\end{align}

The incremental-information approach is well suited to observation models in which the entries of $\bY$ are conditionally independent given $\bX$, that is 
\begin{align}
p_{\bY  \mid \bX}(\by  \mid \bx) & = \prod_{m=1}^M  p_{Y_k \mid \bX}(y_k   \mid \bx) \label{eq:cond_ind}.
\end{align}
The class of models satisfying this condition is quite broad and includes memoryless channels and generalized linear models as special cases. The significance of the conditional independence assumption is summarized in the following result: 

\begin{theorem}[Monotonicity of Incremental Information]\label{thm:mono_inc_I}
The first-order difference sequence $\{ I'_m\}$ is monotonically decreasing for any observation model satisfying the conditional independence condition in \eqref{eq:cond_ind}. 
\end{theorem}
\begin{proof}
Under assumption  \eqref{eq:cond_ind}, two new observations $Y_{\pi(m+1)}$ and $Y_{\pi(m+2)} $ are conditionally independent given $\bX$, and thus the first term on the right-hand side of \eqref{eq:Ipp_alt} is zero.  This means that the second-order difference sequence is non-positive, which implies monotonicity of the first-order difference. 
\end{proof}

The monotonicity in Theorem~\ref{thm:mono_inc_I} can also be seen as a consequence of the subset inequalities studied by Han; see \cite[Chapter~17]{cover:2006}. Our focus on the incremental information is also related to prior work in coding theory that uses an integral-derivative relationship for the mutual information called the area theorem~\cite{measson:2009}.

Similar to the monotonicity properties studied in Section~\ref{sec:I_MMSE}, the monotonicity of the  first-order difference  imposes a number of constraints on the mutual information sequence. Some examples illustrating the usefulness of these constraints  will be provided in the following sections.

\subsection{Standard Linear Model with IID Measurement Vectors}\label{sec:slm_iid}
\index{standard linear model|(} 

We now provide an example of how the incremental-information sequences can be used in the context of the standard linear model~\eqref{eq:linear_model_1}. We focus on the setting where the measurement vectors $\{\bA_m\}$ are drawn iid from a distribution on $\reals^N$. In this setting, the entire observation consists of the pair $(\bY, \bA)$ and the mutual information sequences defined in Section~\ref{sec:inc_mi} can be expressed compactly as
\begin{align*}
I_m & = I(\bX; Y^m \mid  \bA^m) \\
I'_m & = I(\bX; Y_{m+1} \mid Y^m, \bA^{m+1}) \\
I''_m & = - I(Y_{m+1}; Y_{m+2} \mid Y^m, \bA^{m+2}),
\end{align*}
where  $Y^m = (Y_1, \dots, Y_m)$ and $\bA^m = (\bA_1, \dots, \bA_m)$. In these expressions, we do not average over permutations of measurement indices because the distribution of the observations is permutation invariant. Furthermore, the measurement vectors only appear as conditional variables in the mutual information because they are independent of all other random variables.

The sequence perspective can also be applied to other quantities of interest. For example, the MMSE sequence $\{ M_m\}$ is defined by
\begin{align}
M_m & \triangleq   \frac{1}{N} \ex{ \left\| \bX - \ex{ \bX \mid Y^m, A^m} \right\|_2^2},
\label{eq:mmse_seq}
\end{align}
where $M_0 = \frac{1}{N} \gtr(\cov(\bX))$ and $M_M = \frac{1}{N}  \mmse(\bX \mid \bY, \bA)$. By the data-processing inequality for MMSE, it follows that $M_m$ is a decreasing sequence. 

Motivated by the I-MMSE relations in Section~\ref{sec:I_MMSE}, one might wonder if there also exists a relationship between the mutual information and MMSE sequences. For simplicity, consider the setting where the measurement vectors are iid with mean zero and covariance proportional to the identity matrix: 
\begin{align}
 \ex{ \bA_m} = 0 , \qquad \ex{ \bA_m \bA_m^T} =  N^{-1} \Id_N. \label{eq:Am_isotropic} 
\end{align}
One example of a distribution satisfying these constraints is when the entries of $\bA_m$ are iid with mean zero and variance $1/N$. Another example is when $\bA_m$ is drawn uniformly from a collection of $N$  mutually orthogonal unit vectors. 

\begin{theorem}\label{thm:Ip_UB} 
Consider the standard linear model~\eqref{eq:linear_model_1} with iid measurement vectors $\{\bA_m\}$ satisfying~\eqref{eq:Am_isotropic}. If $\bX$ has finite covariance, then the sequences $\{ I'_m\}$ and $\{M_m\}$ satisfy
\begin{align}
I'_m \le \frac{1}{2} \log(1 + M_m) \label{eq:thm:Ip_UB}
\end{align}
for all integers $m$. 
\end{theorem}
\begin{proof}
Conditioned on the observations $(Y^m, \bA^{m+1)}$, the variance of a new measurement can be expressed as
\begin{align*}
\var(\langle \bA_{m+1}, \bX \rangle \mid Y^m , \bA^{m+1}) & =  \bA_{m+1}^T \cov(\bX \mid Y^m, \bA^m)  \bA_{m+1}.
\end{align*}
Taking the expectation of both sides and leveraging the assumptions in~\eqref{eq:Am_isotropic}, we see that
\begin{align}
\ex{ \var(\langle \bA_{m+1}, \bX \rangle \mid Y^m , \bA^{m+1}) } 
& = M_{m} \label{eq:thm:Ip_UB_b}
\end{align}
Next, starting with the fact that mutual-information in a Gaussian channel is maximized when the input (i.e., $\langle \bA_{m+1}, \bX\rangle$) is Gaussian, we have
\begin{align}
I'_m & =  I(\bX; Y_{m+1} \mid Y^m, \bA^{m+1}) \notag \\
& \le \ex{ \frac{1}{2} \log\left( 1 + \var(\langle \bA_{m+1}, \bX \rangle \mid Y^m , \bA^{m+1})  \right)} \notag \\
& \le \ex{ \frac{1}{2} \log\left( 1 + \ex{  \var(\langle \bA_{m+1}, \bX \rangle \mid Y^n , \bA^{m+1})}  \right)}, \label{eq:thm:Ip_UB_d}
\end{align}
where the second step follows from Jensen's inequality and the concavity of the logarithm. Combining~\eqref{eq:thm:Ip_UB_b} and \eqref{eq:thm:Ip_UB_d} gives the stated inequality. 
\end{proof}

Theorem~\ref{thm:Ip_UB} is reminiscent of the I-MMSE relation for Gaussian channels in the sense that it relates a change in mutual information to an MMSE estimate.  One key difference, however, is that~\eqref{eq:thm:Ip_UB} is an inequality instead of an equality. The difference between the right-hand and left-hand sides of~\eqref{eq:thm:Ip_UB} can be viewed as a measure of the difference between the posterior distribution of a new observation $Y_{m+1}$ given observations $(Y^m, \bA^{m+1})$ and the Gaussian distribution with matched first and second moments~\cite{reeves:2016a,reeves:2016,reeves:2016d,reeves:2017c,reeves:2017d}.

Combining Theorem~\ref{thm:Ip_UB}  with the monotonicity of the first-order difference sequence (Theorem~\ref{thm:mono_inc_I}) leads to a lower bound on the MMSE in terms of the total mutual information.   

\begin{theorem}\label{thm:Mm_LB} 
Under the assumptions of Theorem~\ref{thm:Ip_UB}, we have
\begin{align}
M_k \ge \left( \exp\left( \frac{ 2 I_m - k \log(1 + M_0)}{ m - k}    \right) - 1 \right) 
\end{align}
for all integers $0 \le k < m$. 
\end{theorem}
\begin{proof}
For any $0 \le k < m$, the monotonicity of $I'_m$ (Theorem~\ref{thm:mono_inc_I}) allows us to write $I_m = \sum_{\ell = 0}^{m-1} I'_\ell =  \sum_{\ell = 0}^{k-1} I'_k  +  \sum_{\ell = k}^{m-1} I'_\ell \le   k I'_0  + (m-k) I'_k $. Using Theorem~\ref{thm:Ip_UB} to upper bound the terms  $I'_0$ and $I'_k$ and then rearranging terms leads to the stated result. 
\end{proof}

Theorem~\ref{thm:Mm_LB} is particularly meaningful when the mutual information is large. For example if the mutual information satisfies the lower bound
\begin{align*}
 I_{m} \ge (1- \eps) \frac{n}{2} \log(1 + M_0 )  
 \end{align*}
 for some $\eps \in [0, 1)$, then Theorem~\ref{thm:Mm_LB} implies that
 \begin{align*}
 M_k & \ge    \left( \left(1 + M_0\right)^{1 -  \frac{ \eps   }{1 - k   / m}} - 1 \right) 
\end{align*}
for all integers $0 \le k < m$. As $\eps$ converges to zero, the right-hand side of this inequality increases to $M_0$. In other words, a large value of $I_m$ after $m$ observations implies that the MMSE sequence is nearly constant for all $k$ that are sufficiently small relative to $m$.
\index{standard linear model|)} 

%% file: section_proof_of_formulas.tex

\section{Proving the Replica-Symmetric Formula}
\label{sec:proof}

The authors' prior work~\cite{reeves:2016a,reeves:2016,reeves:2016d} provided the first rigorous proof of the replica formulas  \eqref{eq:I_limit_IID} and \eqref{eq:M_limit_IID} for the standard linear model with an iid signal and a Gaussian sensing matrix.

In this section, we give an overview of the proof.
It begins by focusing on the increase in mutual information associated with adding a new observation as described in Sections~\ref{sec:inc_mi} and \ref{sec:slm_iid}. 
Although this approach is developed formally using finite-length sequences, we describe the large-system limit first for simplicity.

\subsection{Large-System Limit and Replica Formulas}
For the large-system limit, the increase in mutual information $\cI(\delta)$ with additional measurement is characterized by its derivative $\cI'(\delta)$.
The main technical challenge is to establish the following relationships:
\begin{align}
 \text{fixed-point formula} &&  \cM(\delta) &= M_{X}\Big( \frac{\delta}{1 + \cM(\delta)} \Big)  && \label{eq:fp} \\
  \text{I-MMSE formula} & &  
 \cI'(\delta)  &=   \frac{1}{2} \log\big(1 + \cM(\delta) \big), \label{eq:I-MMSE} 
\end{align}
where these equalities hold almost everywhere but not at phase transitions.

The next step is to use these two relationships to prove that the replica formulas, \eqref{eq:I_limit_IID} and \eqref{eq:M_limit_IID}, are exact.
First, by solving the minimization over $s$ in~\eqref{eq:I_limit_IID}, one finds that any local minimizer must satisfy the fixed-point formula~\eqref{eq:fp}.
In addition, by differentiating $\cI(\delta)$ in \eqref{eq:I_limit_IID}, one can show that $\cI'(\delta)$ must satisfy the I-MMSE formula~\eqref{eq:I-MMSE}.
Thus, if the fixed-point formula \eqref{eq:fp} defines $\cM(\delta)$ uniquely, then the mutual information $\cI(\delta)$ can be computed by integrating~\eqref{eq:I-MMSE} and the proof is complete.
However, this only happens if there are no phase transitions.
Later, we will discuss how to handle the case of multiple solutions and phase transitions.

\subsection{Information and MMSE Sequences}

To establish the large-limit formulas, \eqref{eq:fp} and \eqref{eq:I-MMSE}, the authors' proof focuses on  functional properties of the incremental mutual information as well as the MMSE sequence $\{ M_m\}$ defined by~\eqref{eq:mmse_seq}.
In particular, the results in \cite{reeves:2016a,reeves:2016} first establish the following approximate relationships between the mutual information and MMSE sequences: 
\begin{align}
 \text{fixed-point formula} && M_{m} &\approx M_X\left(  \frac{m/N}{1 + M_{m}} \right)  && \label{eq:M_fixed_point} \\  
  \text{I-MMSE formula} & &   I'_{m} &\approx \frac{1}{2} \log\left(1 + M_{m}\right). && \label{eq:Ip_MMSE_inq}
\end{align}
The fixed-point formula~\eqref{eq:M_fixed_point} shows that the MMSE $M_m$ corresponds to a scalar estimation problem whose signal-to-noise ratio is a function the number of observations $m$ as well as $M_m$. In Section~\ref{sec:rotation}, it is shown how the standard linear model can be related to a scalar estimation problem of one signal entry. Finally, the I-MMSE formula~\eqref{eq:Ip_MMSE_inq} implies that the increase in information with a new measurement corresponds to a single use of a Gaussian channel with a Gaussian input whose variance is matched to the MMSE.
The following theorem, from \cite{reeves:2016a,reeves:2016}, quantifies the precise sense in which these approximations hold.

\begin{theorem}\label{thm:slm} 
Consider the standard linear model \eqref{eq:linear_model_1} with iid Gaussian measurement vectors $\bA_m \sim \normal(0, N^{-1} \Id_N)$. If the entries of  $\bX$ are iid with bounded fourth moment $\ex{ X_n^4} \le B$, then the sequences $\{I'_m\}$ and $\{M_m\}$ satisfy 
\begin{align}
\sum_{m=1}^{\lceil \delta N \rceil} \left| I'_m - \frac{1}{2} \log(1 + M_m) \right| \le C_{B, \delta}\,  N^{\alpha} \\
\sum_{m=1}^{\lceil \delta N \rceil}  \left| M_m - M_X\left(  \frac{m/N}{1 + M_{m}} \right)  \right| \le C_{B, \delta}\,  N^{\alpha} 
\end{align}
for every integer $N$ and $\delta > 0$ where $\alpha \in (0, 1)$ is a universal constant and $C_{B, \delta}$ is a constant that depends only on the pair $(B, \delta)$. 
\end{theorem}

Theorem~\ref{thm:slm} shows that the normalized sum of the cumulative absolute error in approximations \eqref{eq:M_fixed_point} and \eqref{eq:Ip_MMSE_inq} grows sub-linearly with the vector length $N$. 
Thus, if one normalizes these sums by $M=\delta N$, then the resulting expressions converge to zero as $N\to\infty$.
This is sufficient to establish \eqref{eq:fp} and \eqref{eq:I-MMSE}.

\begin{figure}[t]
\centering
\begin{tikzpicture}
	
	 \node[anchor=south west,inner sep=0, scale = 0.9] at (0,0) {\includegraphics[height=6.5cm]{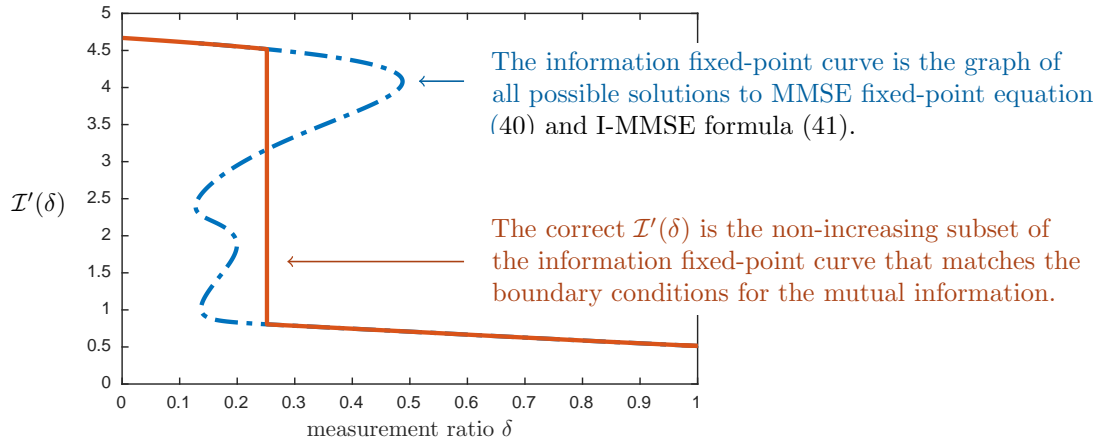}};
	 \node[anchor=south west,inner sep=0, scale = 0.9] at (-3,0) {$\quad$};

	\draw [semithick,  <-,color = {rgb:red,0;green,0.447;blue,0.741}] (4.45, 4.8) -- (5.05,4.8);
        \node[align =left,fill=white, text = {rgb:red,0;green,0.447;blue,0.741}, scale = 1, text width= 8 cm, anchor =west] at (5.3,4.6) {The information fixed-point curve is the graph of \\ all possible solutions to MMSE  fixed-point  equation \\ \eqref{eq:fp} and I-MMSE formula \eqref{eq:I-MMSE}. };
        \draw [semithick,  <-,color = {rgb:red,0.851;green,0.325;blue,0.098}] (2.70, 2.4) -- (5.05,2.4);
        \node[ align =left, fill=white, text = {rgb:red,0.851;green,0.325;blue,0.098}, scale = 1, text width= 8 cm, anchor = west] at (5.3,2.4) {The correct $\cI'(\delta)$ is the non-increasing subset of \\ the information fixed-point curve that matches the \\ boundary conditions for the mutual information.};
        \node[fill=white,text height=0.75cm,text width=3cm] at (6.8,3.6) {};

	

	\node at (-0.6,3.15) [align = left, scale = 0.95]{$\cI'(\delta)$};
	
\end{tikzpicture}
\caption{\label{fig:proof}
Derivative of the mutual information $\cI'(\delta)$ as a function of the measurement rate $\delta$ for linear estimation with iid Gaussian matrices. A phase transition occurs when the derivative jumps from one branch of the information fixed-point curve to another.}
\end{figure}

\subsection{Multiple Fixed-Point Solutions}

At this point, the remaining difficulty is that the MMSE fixed-point formula \eqref{eq:fp} can have \emph{multiple solutions}, as is illustrated by the information fixed-point curve in Figure~\ref{fig:proof}. In this case, the formulas  \eqref{eq:fp} and  \eqref{eq:I-MMSE} alone are not sufficient to uniquely define the actual mutual information $\cI(\delta)$.

For many signal distributions with a single phase transition, the curve (see Figure~\ref{fig:proof}) defined by the fixed-point formula \eqref{eq:fp} has the following property.
For each $\delta$, there are at most two solutions where the slope of the curve is non-increasing.
Since $\cM(\delta)$ is non-increasing, in this case, it must jump from the upper solution branch to the lower solution branch (see Figure~\ref{fig:proof}) at the phase transition.

The final step in the authors' proof technique is to resolve the location of the phase transition using boundary conditions on the mutual information for $\delta = 0$ and $\delta \to \infty$, which can be obtained directly using different arguments. Under the signal property stated below in Definition~\ref{def:single_cross_px}, it is shown that the only solution consistent with the boundary conditions is the one predicted by the replica method. A graphical illustration of this argument is provided in Figure~\ref{fig:proof}.

\subsection{Formal Statement} 

In this section, we formally state the main theorem in~\cite{reeves:2016}.
To do this, we need the following definition.

\begin{definition} \label{def:single_cross_px}
A signal distribution $p_X$ has the \emph{single-crossing property}\footnote{Regrettably, this is unrelated to the ``single-crossing property'' described earlier that says the MMSE function $M_{\bX} (s)$ may cross the matched Gaussian MMSE $M_{\bZ} (s)$ at most once.}
 if its replica MMSE~\eqref{eq:M_limit_IID} crosses its fixed-point curve~\eqref{eq:fp} at most once.
\end{definition}

For any $\delta>0$, consider a sequence of standard linear models indexed by $N$ where the number of measurements is $M=\lceil \delta N \rceil$, the signal $\bX \in \mathbb{R}^N$ is an iid vector with entries drawn from $p_X$, the $M \times N$ measurement matrix $\bA$ has  iid entries drawn from $\mathcal{N}(0,\frac{1}{N})$,
and the observed vector is $\bY = \bA \bX + \bW$ where $\bW \in \mathbb{R}^M$ is a standard Gaussian vector. 
For this sequence, we can also define a sequence of mutual information and MMSE functions:
\begin{align}
\cI_N (\delta) \triangleq &= I(X^N;Y^{\lceil \delta N \rceil}) \\
\cM_N (\delta) \triangleq &= \frac{1}{N} \mmse(X^N \mid Y^{\lceil \delta N \rceil}).
\end{align}

\begin{theorem} \label{thm:limits} Consider the sequence of problems defined above and assume that $p_X$ has a bounded fourth moment (i.e., $\ex{ X^4} \le B < \infty$) and satisfies the single-crossing property.
Then, it follows that:
\begin{enumerate}[(i)]

\item the sequence of mutual information functions $\cI_N (\delta)$ converges to the replica prediction~\eqref{eq:I_limit_IID}.  In other words, for all $\delta > 0$,
 \begin{align}
 \lim_{N \goto \infty}  \cI_N (\delta ) = \cI (\delta). \notag
\end{align}

\item the sequence of MMSE functions $\cM_N (\delta)$ converges almost everywhere to the replica prediction~\eqref{eq:M_limit_IID}.  In other words, at all continuity points of $\cM(\delta)$, 
 \begin{align}
  \lim_{N \goto \infty}  \cM_N (\delta )=  \cM (\delta). \notag
\end{align}

\end{enumerate}

\end{theorem}

\subsubsection{Relationship with Other Methods} 
The use of an integral relationship defining the mutual information is reminiscent of generalized area theorems introduced by M\'{e}asson et al.\ in coding theory~\cite{measson:2009}. However, one of the key differences in the compressed sensing problem is that the conditional entropy of the signal does not drop to zero after the phase transition.

The authors' proof technique also differs from previous approaches that use system-wide interpolation methods to obtain one-sided bounds~\cite{mezard:2009,korada:2010} or that focus on special cases, such as sparse matrices~\cite{montanari:2006}, Gaussian mixture models~\cite{huleihel:2017}, or the detection problem of support recovery~\cite{reeves:2012, reeves:2013b}. After \cite{reeves:2016d,reeves:2016a}, Barbier et al.\ obtained similar results using a substantially different method~\cite{barbier:2016, barbier:2017c}. More recent work has provided rigorous results for the generalized linear model~\cite{barbier:2018a}.

%% file: section_phase_transitions.tex

\section{Phase Transitions and Posterior Correlation} 
\label{sec:phase} 

A \emph{phase transition}  refers to an abrupt change in the macroscopic properties of a system.  In the context of thermodynamic systems, a phase transition may correspond to the transition from one state of matter to another (e.g.,  from solid to liquid or from liquid to gas). In the context of inference problems, a phase transition can be used to describe a sharp change in quality of inference.  For example, the channel coding problem undergoes a phase transition as the signal-to-noise ratio crosses a threshold because the decoder error probability transitions from $\approx \! 1$ to $\approx \! 0$ over a very small range of  signal-to-noise ratios. In the standard linear model, the asymptotic MMSE may also contain a jump discontinuity with respect to the fraction of observations. \index{phase transition}

In many cases, the existence of phase transitions in inference problems can be related to the emergence of significant correlation in the posterior distribution~\cite{mezard:2009}. In these cases, a small change in the uncertainty  for one variable (e.g., reduction in the posterior variance with a new observation) corresponds a change in the uncertainty for a large number of other variables as well. The net effect is a large change in the overall system properties, such as the MMSE. 

In this section, we show how the tools introduced in Section~\ref{sec:role_of_MI} can be used to provide a link between a measure of the average correlation in the posterior distribution and second-order differences in the mutual information for both the Gaussian channel and the standard linear model.  

\subsection{Mean-Squared Covariance} \label{sec:msc}
\index{mean-squared covariance}
Let us return to the general inference problem of estimating a random vector $\bX =(X_1, \dots, X_N)$ from observations $\bY = (Y_1, \dots, Y_M)$.  As discussed in Section~\ref{sec:role_of_MI}, the posterior covariance matrix $\cov(\bX \mid \bY )$  provides a geometric measure of the amount of uncertainty in the posterior distribution. One important function of this  matrix is the  MMSE, which corresponds to the expected posterior variance:
\begin{align}
\mmse(\bX  \mid \bY) = \sum_{n=1}^N \ex{ \var(X_n  \mid \bY)}.
\end{align}

Going beyond the MMSE, there is also important information contained in the off-diagonal entries, which describe the pairwise correlations.  A useful measure of this correlation is provided by the \emph{mean-squared covariance}: 
\begin{align}
\ex{ \left\|\cov(\bX \mid \bY) \right \|_F^2} = \sum_{k=1}^N \sum_{n=1}^N \ex{ \cov^2(X_k, X_n \mid \bY)} . \label{eq:msc}
\end{align}
Note that while the MMSE corresponds to $N$ terms, the mean-squared covariance corresponds to $N^2$ terms. If the entries in $\bX$ have bounded fourth moments (i.e.,  $\ex{ X_i^4} \le B$), sthen it follows from the Cauchy-Schwarz inequality that each summand on the right-hand side of \eqref{eq:msc} is upper bounded by $B$, and it can be verified that
\begin{align}
\frac{1 }{N} \left( \mmse(\bX  \mid \bY)  \right)^2\le   \ex{ \left\|\cov(\bX \mid \bY) \right \|_F^2} \le  N^2 B. \label{eq:msc_bnds}
\end{align}
The left inequality is tight when the posterior distribution is uncorrelated and hence the off-diagonal terms of the conditional covariance are zero.  The right inequality is tight when the off-diagonal terms are of the same order as the variance.

Another way to view the relationship between the MMSE and mean-squared covariance it to consider the spectral decomposition of the covariance matrix. Let $\Lambda_1 \ge \Lambda_2 \ge \dots \ge \Lambda_N$ denote the (random) eigenvalues of  $\cov(\bX \mid \bY)$.  Then we can write
\begin{align*}
\gtr( \cov(\bX \mid \bY)) &= \sum_{n =1}^N   \Lambda_n, \qquad 
\left\|\cov(\bX \mid \bY) \right\|_F^2=     \sum_{n =1}^N  \Lambda_n^2.
\end{align*} 
Taking the expectations of these random quantaties and rearranging terms, 
one finds that the mean-squared covariance can be decomposed into three nonnegative terms: 
\begin{align}
 \ex{ \left\|\cov(\bX \mid \bY) \right \|_F^2} 
&=  N \left(  \ex{ \bar{\Lambda}} \right)^2 + N  \var( \bar{\Lambda}) +    \sum_{n=1}^N \ex{ \left( \Lambda_n - \bar{\Lambda} \right)^2}  , \label{eq:cov_decomp}
\end{align}
where $\bar{\Lambda} = \frac{1}{N} \sum_{n=1}^N \Lambda_n $ denotes the arithmetic  mean of the eigenvalues.  The first term on the right-hand side corresponds to the square of the MMSE and is equal to the lower bound in \eqref{eq:msc_bnds}. The second term on the right-hand side corresponds to the variance of  $\frac{1}{N}\gtr(\cov(\bX \mid \bY))$ with respect to the randomness in $\bY$. This term is equal to zero if $\bX$ and $\bY$ are jointly Gaussian.  The last term corresponds to the expected variation in the eigenvalues. If a small number of eigenvalues are significantly larger than the others then  it is possible for this term to be $N$ times larger than the first term. When this occurs, most of the uncertainty in the posterior distribution is concentrated on a low-dimensional subspace.

\subsection{Conditional MMSE Function and its Derivative}

The relationship between phase transitions and correlation in the posterior distribution can be made precise using the properties of the Gaussian channel discussed in Section~\ref{sec:I_MMSE}.  Given a random vector $\bX \in \mathbb{R}^N$ and a random observation $\bY \in \mathbb{R}^M$,  the \emph{conditional MMSE function} is defined by
\begin{align}
M_{\bX \mid \bY}(s) & \triangleq \frac{1}{N} \ex{ \| \bX - \ex{\bX \mid \bY, \bZ(s)}\|^2}, \label{eq:M_cond}
\end{align}
where  $\bZ (s) = \sqrt{s} \, \bX +\bW$ is a new observation of $\bX$ from an independent Gaussian noise channel~\cite{reeves:2017e}. From the expression for the derivative of the MMSE in \eqref{eq:Mprime}, it can be verified that the derivative of the conditional  MMSE function is  
\begin{align}
M'_{\bX \mid \bY}(s) &=  - \frac{1}{N} \ex{\left \|  \cov(\bX \mid \bY,  \bZ (s) ) \right\|_F^2}. \label{eq:Mp_cond}
\end{align}
Here, we recognize that the right-hand side is proportional to the mean-squared covariance associated with the pair of observations $(\bY, \bZ(s))$. Meanwhile, the left-hand side describes the change in the MMSE associated with a small increase in the signal-to-noise ratio of the Gaussian channel.

In Section~\ref{sec:good_codes} we saw that a phase transition in the channel coding problem for the Gaussian channel corresponds to a jump discontinuity the MMSE. More generally, one can say that the inference problem defined by the pair $(\bY, \bZ(s))$ undergoes a phase transition  whenever  $M_{\bX \mid \bY}(s)$ has a jump discontinuity in the large-$N$ limit. If such a phase transition occurs, then it implies that the magnitude of $M'_{\bX \mid \bY}(s)$ is increasing without bound. From \eqref{eq:Mp_cond}, we see that this also implies significant correlation in the posterior distribution. 

Evaluating the conditional MMSE function and its derivative at $s = 0$ provides expressions for the MMSE and mean-squared covariance associated with the orignal observation model:
\begin{align}
M_{\bX \mid \bY}(0) &= \frac{1}{N} \mmse(\bX  \mid \bY) \\
M'_{\bX \mid \bY}(0) &=  - \frac{1}{N} \ex{\left \|\cov(\bX \mid \bY) \right \|_F^2}.
\end{align}
In light of the discussion above,  the mean-squared covariance can be interpreted as the rate of MMSE change with $s$ that occurs when one is presented with an independent observation of $\bX$ from a Gaussian channel with infinitesimally small signal-to-noise ratio. Furthermore, we see that significant correlation in the posterior distribution corresponds to a jump discontinuity in the large-$N$ limit of  $M_{\bX \mid \bY}(s)$ at the point $s  = 0$.

\subsection{Second-Order Differences of the Information Sequence}

Next, we  consider some further properties of the incremental-mutual information sequence introduced in Section~\ref{sec:inc_mi}. For any observation model satisfying the conditional independence condition in \eqref{eq:cond_ind} the second-order difference sequence given in \eqref{eq:Ipp_alt} can be expressed as
\begin{align}
I''_{m} & =   - \frac{1}{M!}  \sum_{\pi }  I(Y_{\pi(m+2)}; Y_{\pi(m+1)} \mid Y_{\pi(1)}, \dots Y_{\pi(m)}) . 
\end{align}
Note that each summand is a measure of the \emph{pairwise dependence} in the posterior distribution of new measurements. If the pairwise dependance is large (on average) then it means that there is a significant decrease in the first-order difference sequence $I'_m$.

The monotonicity and  non-negativity of the first-order difference  sequence  imposes some important constraints on the second-order difference sequence. For example,  the number of terms for which $|I''_m|$ is ``large'' can be upper bounded in terms of the information provided by a single observation.  The following result provides a quantitative description of this constraint.  

\begin{theorem}\label{thm:card_bnd} 
For any observation model satisfying the conditional independence condition in \eqref{eq:cond_ind} and positive number  $T$, the second-order difference sequence $\{ I''_m\}$ satisfies 
\begin{align}
| \left\{ m  \, : \,   |I''_m| \ge  T  \right \}|  & \le I_1 / T ,\label{eq:card_bnd} 
\end{align}
where $I_1 = \frac{1}{M} \sum_{m=1}^M I(\bX ; Y_m)$ is the first term in the information sequence. 
\end{theorem}
\begin{proof}
The monotonicity of the first-order difference (Theorem~\ref{thm:mono_inc_I}) means that $I''_m$ is non-positive, and hence the indicator function of the event $\{|I''_m | \ge T\}$ is upper bounded by $-  I''_m /T$. Summing this inequality over $m$, we obtain
\begin{align*}
|\left\{ m  \, : \,   |I''_m| \ge  T  \right\}|  = \sum_{m=1}^{M-2} \one_{[T,\infty)} ( | I''_m| )  \le    \sum_{m=1}^{M-2}  - I''_m /T  =  \left( I'_0 - I'_{M-1} \right)/T .
\end{align*}
Noting that $I'_0= I_1$ and $I'_{M-1} \ge 0$ completes the proof. 
\end{proof}

An important property of Theorem~\ref{thm:card_bnd} is that for many problems of interest, the term $I_1$ does not depend on the total number of observations $M$. For example, in the standard linear model with iid measurement vectors, the upper bound in Theorem~\ref{thm:Ip_UB} gives  $I_1 \le \frac{1}{2} \log(1 + M_{\bX}(0))$. Some implications of these results are discussed in the next section. 

\subsubsection{Implications for the Standard Linear Model} 
One of the key steps in the authors' prior work on the standard linear model \cite{reeves:2016a,reeves:2016,reeves:2016d} is  the following inequality, which relates the second-order difference sequence and the mean-squared covariance.

\begin{theorem}\label{thm:cov_to_Ipp}
Consider the standard linear model \eqref{eq:linear_model_1} with iid Gaussian measurement vectors $\bA_m \sim \normal(0, N^{-1} \Id_N)$. If the entries of  $\bX$ are independent with bounded fourth-moment $\ex{ X_n^4} \le B$, then the mean-squared covariance satisfies
\begin{align}
\frac{1}{ N^2} \ex{ \left\| \cov(\bX \mid Y^m , \bA^m ) \right\|_F^2 }  \le C_B \,    \left| I''_m \right|^{1/4}, \label{eq:Ipp_one_fourth}
\end{align}
for all integers $N$ and $m = 1, \dots, M$ where $C_B$ is a constant that depends only on the fourth-moment upper bound $B$. 
\end{theorem}

Theorem~\ref{thm:cov_to_Ipp} shows that significant correlation in the posterior distribution implies pairwise dependence in the joint distribution of new measurements and, hence, a significant decrease in the first-order difference sequence $I'_m$. In particular, if the mean-squared covariance is order $N^2$ (corresponding to the upper bound in \eqref{eq:msc_bnds}), then the $|I''_m|$ is lower bounded by a constant. If we consider the large-$N$ limit in which the number of observations is parameterized by the fraction $\delta = m/N$, then an order one difference in $I'_m$ corresponds to a jump-discontinuity with respect to $\delta$. In other words, significant pairwise correlation implies a phase transition with respect to the fraction of observations. 

Viewed in the other direction, Theorem~\ref{thm:cov_to_Ipp} also shows that small changes in the first-order difference sequence imply that the average pairwise correlation is small. From Theorem~\ref{thm:card_bnd}, we see that this is, in fact, the typical situation. Under the assumptions of Theorem~\ref{thm:cov_to_Ipp},  it can be verified that 
\begin{align}
\left| \left\{ m  \, : \,  \ex{ \left\| \cov(\bX \mid Y^m , \bA^m ) \right\|_F^2 } \ge  N^{2-\eps/4}     \right \} \right|  & \le  \tilde{C}_B\, N^{\eps} 
\end{align}
for all $0 \le \eps \le 1$, where $\tilde{C}_B$ is a constant that depends only on the fourth-moment bound $B$.  In other words, the number of $m$-values for which the mean-squared covariance has the same order as the upper bound in \eqref{eq:msc_bnds}  must be sub-linear  in  $N$. This fact plays a key role in the proof of Theorem~\ref{thm:slm}; see \cite{reeves:2016a,reeves:2016} for details.

%% file: section_extensions.tex

\section{Appendix: Subset Response for the Standard Linear Model} 

\label{sec:rotation}

This section describes the mapping between the standard linear model and a signal-plus-noise response model for a subset of the observations. Recall the problem formulation
\begin{align}
\bY = \bA \bX + \bW
\end{align}
where $\bA$ is an $M \times N$ matrix. Suppose that we are interested in the posterior distribution of a subset   $S \subset \{1,\dots, N\}$ of the signal entries where the size of the subset $K = |S|$ is small relative to the signal length $N$ and the number of measurements $M$. Letting $S^c = \{1,\dots, N\} \backslash S$ denote the complement of $S$, the measurements can be decomposed as 
\begin{align}
\bY & = \bA_{S} \bX_{S}  + \bA_{S^c} \bX_{S^c}   + \bW , 
\end{align}
where $\bA_{S}$ is an $M \times K$ matrix corresponding to the columns of $\bA$ indexed by $S$ and $\bA_{S^c}$ is an $M \times (N-K)$ matrix corresponding to the columns indexed by the compliment of $S$.

This decomposition suggests an alternative interpretation of the linear model in which $\bX_S$ is a low-dimensional signal of interest and $\bX_{S^c}$ is a high-dimensional interference term. Note that  $\bA_{S}$ is a tall skinny matrix, and thus the noiseless measurements of $\bX_S$  lie in a $K$-dimensional subspace of the $M$-dimensional measurement space.

Next, we introduce a linear transformation of the problem that attempts to separate the signal of interest from the interference term. The idea is to consider the QR decomposition of the tall skinny matrix $\bA_{S}$ of the form
\[
\bA_S 
= \underbrace{ \begin{bmatrix} \bQ_1 ,  \bQ_2  \end{bmatrix} }_{\bQ} \begin{bmatrix} \bR \\  \bm{0} \end{bmatrix},
\]
where $\bQ$ is an $M \times M$  orthogonal matrix ($\bQ \bQ^T = I$), $\bQ_1$ is $M \times K$, $\bQ_2$ is $M \times (M-K)$,  and $\bR$ is  an $K \times K$ upper triangular matrix whose diagonal entries are nonnegative. If $\bA_S$ has full column rank then the pair $(\bQ_1, \bR)$ is uniquely defined. The matrix  $\bQ_2$ can be chosen arbitrarily subject to the constraint $\bQ_2^T \bQ_2 = I - \bQ^T_1 \bQ_1$. To facilitate the analysis, we will assume that  $\bQ_2$ is chosen  uniformly at random over the set of matrices satisfying this constraint. 

Multiplication by $\bQ^T$ is a one-to-one linear transformation. The transformed  problem parameters are defined as:
\[
\tilde{\bY} \triangleq \bQ^T \bY , \qquad \bB \triangleq \bQ^T \bA_{S^c} , \qquad  \tilde{\bW} \triangleq \bQ^T \bW.
\]
At this point, it is important to note that the isotropic Gaussian distribution is invariant to orthogonal transformations. Consequently, the transformed noise $\tilde{\bW}$  has the same distribution as $\bW$ and is independent of everything else. Using the transformed parameters, the linear model can be expressed as
\begin{align}
\begin{bmatrix}
\tilde{\bY}_1\\
\tilde{\bY}_2
\end{bmatrix} 
& = \begin{bmatrix} \bR & \bB_{1} \\
\bm{0} & \bB_{2} 
\end{bmatrix} 
\begin{bmatrix}
\bX_{S}\\
\bX_{S^c}
\end{bmatrix} 
+ 
\begin{bmatrix}
\tilde{\bW}_1\\
\tilde{\bW}_2
\end{bmatrix},
\end{align}
where $\tilde{\bY}_1$ corresponds to the first $K$ measurements and $\tilde{\bY}_2$ corresponds to the remaining $(M-K)$ measurements. 

A useful property of the transformed model is that  $\tilde{\bY}_2$ is independent of the signal of interest $\bX_S$. This decomposition  motivates a  two-stage approach in which one first estimates $\bX_{S^c}$ from the data $(\tilde{\bY}_2, \bB_2)$ and then uses this estimate to ``subtract out'' the interference term in $\tilde{\bY}_1$. To be more precise, we define
\begin{align*}
\bZ 
 & \triangleq  \tilde{\bY}_1 -  \bB_1 \ex{ \tilde{\bX}_{S^c}   \mid \tilde{\bY}_2 , \bB_2},
\end{align*}
to be the measurements $\tilde{\bY}_1$ after subtracting the conditional expectation of the interference term. Rearranging terms, one finds that the relationship between $\bZ$ and $\bX_S$ can be expressed succinctly as
\begin{align}
\bZ & = \bR \bX_{S} + \bV , \quad \bV \sim p(\bv \mid \tilde{\bY}_2, \bB) ,
\end{align}
where 
\begin{align}
\bV & \triangleq \bB_{1}\left(  \bX_{S^c} - \ex{ \bX_{S^c}  \mid \tilde{\bY}_2, \bB_{2} } \right) + \tilde{\bW}_1, 
\end{align}
is the error due to both the interference and the measurement noise.

Thus far, this decomposition is quite general in the sense that it can be applied for any matrix $\bA$ and subset $S$ of size less than $M$. The key question at this point is whether the error term $\bV$ is approximately  Gaussian.